\documentclass[11pt,letterpaper]{article}
\usepackage{hyperref}
\usepackage{epsfig,amssymb,amsmath,color,enumerate,epstopdf}
\usepackage{afterpage,amsthm}
\usepackage{algorithm,algorithmic}
\usepackage{algorithm,algorithmic}

\newcommand{\bmat}{\left[ \begin{array}}
\newcommand{\emat}{\end{array} \right]}

\usepackage[headheight=-1in,margin=1in]{geometry}

\newtheorem{thm}{Theorem} 
\newtheorem{theorem}[thm]{Theorem}

\newtheorem{lemma}[thm]{Lemma}
\newtheorem{corollary}[thm]{Corollary}

\newcommand{\remove}[1]{}

\newcommand{\lt}{\left}
\newcommand{\rt}{\right}

\newcommand{\ignore}[1]{}

\renewcommand{\algorithmicforall}{\textbf{parallel for each}}

\renewcommand{\algorithmicrequire}{\textbf{Input:}}
\renewcommand{\algorithmicensure}{\textbf{Output:}}

\begin{document}
\title{Minimizing Communication in Linear Algebra}

\author{
Grey Ballard \and
James Demmel \and
Olga Holtz \and
Oded Schwartz
}



\maketitle

\begin{abstract}
In 1981 Hong and Kung \cite{HongKung81} proved a lower bound on the amount
of communication (amount of data moved between a small, fast memory
and large, slow memory) needed to perform dense, $n$-by-$n$
matrix-multiplication using the conventional $O(n^3)$ algorithm,
where the input matrices were too large to fit in the small, fast memory.
In 2004 Irony, Toledo and Tiskin \cite{IronyToledoTiskin04}
gave a new proof of this
result and extended it to the parallel case (where communication
means the amount of data moved between processors).
In both cases the lower bound may be expressed as
$\Omega$(\#arithmetic operations / $\sqrt{M}$), where
$M$ is the size of the fast memory (or local memory in
the parallel case).

Here we generalize these results to a much wider variety of
algorithms, including LU factorization, Cholesky factorization,
$LDL^T$ factorization, QR factorization, algorithms for
eigenvalues and singular values, i.e., essentially all direct
methods of linear algebra.

The proof works for dense or sparse matrices,
and for sequential or parallel algorithms.
In addition to lower bounds on the amount of data moved (bandwidth)
we get lower bounds on the number of messages required to move it (latency).

We illustrate how to extend our lower bound technique
to compositions of linear algebra operations
(like computing powers of a matrix), to decide whether it is enough
to call a sequence of simpler optimal algorithms (like matrix multiplication)
to minimize communication, or if we can do better. We give examples of both.
We also show how to extend our lower bounds to certain graph theoretic problems.

We point out recently designed algorithms for dense LU, Cholesky,
QR, eigenvalue and the SVD problems that attain these lower
bounds; implementations of LU and QR show large speedups over
conventional linear algebra algorithms in standard libraries like
LAPACK and ScaLAPACK.  Many open problems remain.
\end{abstract}

\section{Introduction}

Algorithms have two kinds of costs: arithmetic and communication,
by which we mean either moving data between levels of a memory
hierarchy (in the sequential case) or over a network connecting
processors (in the parallel case). There are two costs associated with
communication: {\em bandwidth} (proportional to the total number of words of
data moved) and {\em latency} (proportional to the number of messages in
which these words are packed and sent).
For example, we may model the cost of sending $m$ words in a
single message as $\alpha + \beta m$, where $\alpha$
is the latency (measured in seconds) and $\beta$ is the reciprocal
bandwidth (measured in seconds per word). Depending on the technology,
either latency or bandwidth costs may be larger,
often dominating the cost of arithmetic.
So it is of interest to have algorithms minimizing both communication costs.

In this paper we prove a general lower bound on the amount of data
moved (i.e., bandwidth) by a general class of algorithms,
including most dense and sparse linear algebra algorithms, as well
as some graph theoretic algorithms. Our model is the result of
Hong and Kung \cite{HongKung81} which says that to multiply two
dense $n$-by-$n$ matrices on a machine with a large slow memory
(in which the matrices initially reside) and a small fast memory
of size $M$ (too small to store the matrices, but arithmetic may
only be done on data in fast memory), $\Omega(n^3/\sqrt{M})$ words
of data must be moved between fast and slow memory. This lower
bound is attained by a variety of ``blocked'' algorithms. This
lower bound may also be expressed as
$\Omega$(\#arithmetic\_operations / $\sqrt{M}$) \footnote{The
sequential communication model used here is sometimes called the
\emph{two-level I/O model} or \emph{disk access machine (DAM)}
model (see \cite{AggarwalVitter88},
\cite{BenderBrodalFagerbergJacobVicari07},
\cite{ChowdhuryRamachandran06}).  Our model follows that of
\cite{HongKung81} and \cite{IronyToledoTiskin04} in that it
assumes the block-transfer size is one word of data ($B=1$ in the
common notation).  }.

This result was proven differently by Irony, Toledo and Tiskin
\cite{IronyToledoTiskin04} and generalized  to the parallel case,
where $P$ processors multiply two $n$-by-$n$ matrices. In the
``memory-scalable'' case, where each processor stores the minimal
$O(n^2/P)$ words of data, they obtain the lower bound
$\Omega$(\#arithmetic\_operations\_per\_processor / $\sqrt{{\rm
memory\_per\_processor}})$ = $\Omega( \frac{n^3/P}{\sqrt{n^2/P}} )
= \Omega( \frac{n^2}{\sqrt{P}} )$, which is attained by Cannon's
algorithm \cite{Cannon69} \cite[Lecture 11]{CS267a}. The paper
\cite{IronyToledoTiskin04} also considers the ``3D'' case, which
does less communication by replicating the matrices and so using
$O(P^{1/3})$ times as much memory as the minimal possible.

Here we begin with the proof in
\cite{IronyToledoTiskin04}, which starts with the sum $C_{ij} =
\sum_k A_{ik}\cdot B_{kj}$, and uses a geometric argument on the
lattice of indices $(i,j,k)$ to bound the number of updates
$C_{ij} += A_{ik} \cdot B_{kj}$ that can be performed when a
subset of matrix entries are in fast memory. This proof
generalizes in a number of ways: in particular it does not
depend on the matrices being dense, or the output being
distinct from the input.
These observations let us state and prove a general
Theorem~\ref{thm:Main} in section~\ref{sec:MainResult},
that a lower bound on the number of words moved into or out of
a fast or local memory of size $M$ is
$\Omega$(\#arithmetic operations / $\sqrt{M}$ ).
This applies to both the sequential case (where $M$
is a fast memory) and the parallel case; in the parallel
case further assumptions about whether the algorithm
is memory balanced (to estimate the effective $M$)
are needed to get a lower bound on the overall algorithm.

Corollary~\ref{cor:Latency} of Theorem~\ref{thm:Main} provides a
simple lower bound on latency (just the lower bound on bandwidth
divided by the largest possible message size, namely the memory
size $M$). Both bandwidth and latency lower bounds apply
straightforwardly to a nested memory hierarchy with more than two
layers, bounding from below the communication between any adjacent
layers in the hierarchy
\cite{Savage95,BallardDemmelHoltzSchwartz09}.

In Section~\ref{sec:Consequences}, we present simple corollaries
applying Theorem~\ref{thm:Main} to conventional
(non-Strassen-like) implementations of matrix multiplication and
other BLAS operations \cite{BLAS,newBLAS,newBLASfull} (dense or
sparse), LU factorization, Cholesky factorization and ``$LDL^T$''
factorization. These factorizations may also be dense or sparse,
with any kind of pivoting, and be exact or ``incomplete'', e.g.,
ILU \cite{Saad96} (some of these results can be also obtained,
just for dense matrices, by suitable reductions from
\cite{HongKung81} or \cite{IronyToledoTiskin04}, and we point
these out).

Section~\ref{sec:QR} considers lower bounds for algorithms that
apply orthogonal transformations to the left and/or right of
matrices. This class includes the QR factorization, the
standard algorithms for eigenvalues and eigenvectors, and the
singular value decomposition (SVD). For reasons explained there,
the counting techniques of \cite{HongKung81} and
\cite{IronyToledoTiskin04} do not apply, so we need a different
but related lower bound argument.

Section~\ref{sec:Compositions} shows how to extend our lower
bounds to more general computations where we compose a sequence of
simpler linear algebra operations (like matrix multiplication, LU
decomposition, etc.), so the outputs of one operation may be
inputs to later ones. If these intermediate results do not need to
be saved in slow memory, or if some inputs are given by formulas
(like $A(i,j) = 1/(i+j)$) and so do not need to be fetched from
memory, or if the final output is just a scalar (the norm or
determinant of a matrix), then it is natural to ask whether there
is a better algorithm than just using optimized versions of each
operation in the sequence. We give examples where this simple
approach is optimal, and when it is not. We also exploit the
natural correspondence between matrices and graphs to derive
communication lower bounds for certain graph algorithms, like
All-Pairs-Shortest-Path.

Finally, Section~\ref{sec:open} discusses attainability of these
lower bounds, and open problems. Briefly, in the dense case all
the lower bounds are attainable (in the parallel case, this is
modulo ${\rm polylog} P$ factors, and assuming the minimal
$O(n^2/P)$ storage per processor); see Tables~\ref{tbl:seq} and
\ref{tbl:par} (some of these algorithms are also pointed out in
sections~\ref{sec:Consequences} and \ref{sec:QR}). The optimal
algorithms for square matrix multiplication are well known, as
mentioned above. Optimal algorithms for dense LU, Cholesky, QR,
eigenvalue problems and the SVD are more recent, and not part of
standard libraries like LAPACK \cite{LAPACK} and ScaLAPACK
\cite{SCALAPACK}. Only in the case of Cholesky do we know of a
sequential algorithm that both minimizes bandwidth and latency
across arbitrary levels of memory hierarchy. No optimal algorithm
is known for architecture mixing parallelism and multiple memory
hierarchies, i.e., most real architectures. Optimal ``3D''
algorithms for anything other than matrix-multiplication, and
optimal sparse algorithms for anything are unknown. For highly
rectangular dense matrices (e.g., matrix-vector multiplication),
or for sufficiently sparse matrices, our new lower bound is
sometimes lower than the trivial lower bound (\#inputs +
\#outputs), and so it is not always attainable.





\section{First Lower Bound}
\label{sec:MainResult}

Let $c(i,j)$ be the memory address of a destination to put a
computed result, where $i$ and $j$ are integer indices (thus
$Mem(c(i,j))$ will contain the actual value). All we assume is
that $c(i,j)$ is a one-to-one function on a set $S_C$ of pairs
$(i,j)$ for destinations we want to compute; in other words all
$c(i,j)$ for $(i,j) \in S_C$ are distinct.
On a parallel machine $c(i,j)$ refers to
a location on some processor; the processor number
is implicitly part of $c(i,j)$.

Similarly, let $a(i,k)$ and $b(k,j)$ be memory
addresses of operands, also one-to-one on sets
made explicit below. We make no assumptions about
whether or not some $c(i,j)$ can ever equal some $a(i',j')$;
they may or may not. Similarly, the addresses
given by $a$ and $b$, or by $b$ and $c$, may
overlap arbitrarily.
We assume that the result to be stored at
each $c(i,j)$ is computed only once.

Now let $f_{ij}$ and $g_{ijk}$ be ``nontrivial'' functions
in a sense we make clear below.
The computation we want to perform is for all $(i,j) \in S_C$:
\begin{equation}
\label{eqn:Model}
Mem(c(i,j)) = f_{ij} ( g_{ijk} (Mem(a(i,k)),Mem(b(k,j)))
\; {\rm for} \; k \in S_{ij} , {\rm any\ other\ arguments} )
\end{equation}
Here $f_{ij}$ depends nontrivially on its arguments $g_{ijk}(
\cdot , \cdot )$ which in turn depend nontrivially on their
arguments $Mem(a(i,k))$ and $Mem(b(k,j))$, in the following sense:
we need at least one word of space to compute $f_{ij}$ (which may
or may not be $Mem(c(i,j))$) to act as ``accumulator'' of the
value of $f_{ij}$, and we need the values $Mem(a(i,k))$ and
$Mem(b(k,j))$ in fast memory before evaluating $g_{ijk}$. Note
also that we may not know until after the computation what $S_C$,
$f_{ij}$, $S_{ij}$, $g_{ijk}$ or ``any other arguments'' were,
since they may be determined on the fly (e.g., pivot order).

The question is how many slow memory references are required to
perform this computation, when all we are allowed to do is compute
the $g_{ijk}$ in a different order, and compute and store the
$f_{ij}$ is a different order. This appears to restrict possible
reorderings to those where $f_{ij}$ is computed correctly, since
we are not assuming it is an associative or commutative function,
or those reorderings that avoid races because some $c(i,j)$ may be
used later as inputs. But there is no need for such restrictions:
the lower bound applies to all reorderings, correct or incorrect.
Using only structural information, e.g., about the sparsity
patterns of the matrices, we can sometimes deduce that the
computed result $f_{ij}( \cdot )$ is exactly zero, to possibly
avoid a memory reference to store the result at $c(i,j)$. Section
\ref{sec:careful} discusses this possibility more carefully, and
shows how to carefully count operations to preserve the validity
of our lower bounds.

The argument, following \cite{IronyToledoTiskin04}, is: (1) Break
the stream of instructions executed into segments, where each
segment contains exactly $M$ load and store instructions (i.e.,
that cause communication), where $M$ is the fast (or local) memory
size. (2) Bound from above the number of evaluations of functions
$g_{ijk}$ that can be performed during any segment, calling this
upper bound $F$. (3) Bound from below the number of (complete)
segments by the total number of evaluations of $g_{ijk}$ (call it
$G$) divided by $F$, i.e., $\lfloor G/F \rfloor$. (4) Bound from
below the number of loads and stores by $M$ times the minimum
number of complete segments, $M \cdot \lfloor G/F \rfloor$.

Now we compute the upper bound $F$ using a geometric theorem
of Loomis and Whitney \cite{LoomisWhitney49,BuragoZalgaller88}.
We need only the simplest version of their result here:

\begin{lemma}
\label{lemma:LW} \cite{LoomisWhitney49,BuragoZalgaller88}. Let $V$
be a finite set of lattice points in ${\bf R}^3$, i.e., points
$(x,y,z)$ with integer coordinates. Let $A_x$ be the projection of
$V$ in the $x$-direction, i.e., all points $(y,z)$ such that there
exists an $x$ so that $(x,y,z) \in V$. Define $A_y$ and $A_z$
similarly. Let $| \cdot |$ denote the cardinality of a set. Then
$|V| \leq \sqrt{ |A_x| \cdot |A_y| \cdot |A_z| }$.
\end{lemma}

Now we must bound the maximum number of possibly different
$Mem(c(i,j))$ (or corresponding ``accumulators''), $Mem(a(i,k))$,
and $Mem(b(k,j))$ that can reside in fast memory during a segment.
Since we want to accommodate the most general case where input and
output arguments can overlap, we need to use a more complicated
model than in \cite{IronyToledoTiskin04}, where no such overlap
was possible. To this end, we consider each input or output
operand of (\ref{eqn:Model}) that appears in fast memory during a
segment of $M$ slow memory operations. It may be that an operand
appears in fast memory for a while, disappears, and reappears,
possibly several times (we assume there is at most one copy at a
time in the sequential model, and at most one for each processor
in the parallel model; this obviously is consistent with obtaining
a lower bound). For each period of continuous existence of an
operand in fast memory, we label its Source (how it came to be in
fast memory) and its Destination (what happens when it
disappears):
\begin{itemize}
\item
{\bf Source S1:} The operand was already in fast memory at the beginning of the
segment, and/or read from slow memory.
There are at most $2M$ such operands altogether,
because the fast memory has size $M$, and because a segment
contains at most $M$ reads from slow memory.
\item
{\bf Source S2:} The operand is computed (created) during the segment.
Without more information,
there is no bound on the number of such operands.
\item
{\bf Destination D1:} An operand is left in fast memory at the end of the
segment (so that it is available at the beginning of the next one),
and/or written to slow memory.
There are at most $2M$ such operands altogether,
again because the fast memory has size $M$, and because a segment
contains at most $M$ writes to slow memory.
\item
{\bf Destination D2:} An operand is {\em neither} left in fast memory nor written to
slow memory, but simply discarded.
Again, without more information,
there is no bound on the number of such operands.
\end{itemize}

We may correspondingly label each period of continuous existence
of any operand in fast memory during one segment by one of four
possible labels Si/Dj, indicating the Source and Destination of
the operand at the beginning and end of the period. Based on the
above description, the total number of operands of all types
except S2/D2 is bounded by $4M$ (the maximum number of S1 operands
plus the number of D1 operands, an upper bound) \footnote{More
careful but complicated accounting can reduce this upper bound to
$3M$.}. The S2/D2 operands, those created during the segment and
then discarded without causing any slow memory traffic, cannot be
bounded without further information. For our simplest model,
adequate for matrix multiplication, LU decomposition, etc., we
have no S2/D2 arguments; they reappear when we analyze the QR
decomposition in Section~\ref{sec:QR}.

Using the set of lattice points $(i,j,k)$
to represent each function evaluation
\linebreak
$g_{ijk} (Mem(a(i,k)),Mem(b(k,j)))$, and assuming there are no S2/D2
arguments, then by Lemma~\ref{lemma:LW}
their number is then bounded by $F = \sqrt{(4M)^3}$,
so the total number of loads and stores
is bounded by
$M \lfloor \frac{G}{F} \rfloor = M \lfloor \frac{G}{\sqrt{(4M)^3}} \rfloor \geq
\frac{G}{8\sqrt{M}} - M$.
This proves the first lower bound:

\ignore{
--------------------------------------------------------
Now we must bound the maximum number of possibly different
$Mem(c(i,j))$ (or corresponding ``accumulators'') that can reside
in fast memory during a segment.
There are three ways such locations can be in fast memory:
they can have been there at the beginning of the segment
(at most $M$),
they can be read in from memory
(at most $M$),
or they can be computed for the first time during the segment.
In the latter case, at most $2M$ can be created, because
the first $M$ will have to be written to slow memory to
make room for the second $M$.
Considering an arbitrary sequence of $M$ memory operations,
at most $2M$ different locations $c(i,j)$ can possibly have
been accessed: $M$ (updated) values that are written to
memory, and $M$ values that are left in memory for the next segment
to continue accessing.
The same is true for $Mem(a(i,k))$ and $Mem(b(k,j))$,
whether these locations coincide or not.

Using the set of lattice points $(i,j,k)$
to represent each function evaluation
\linebreak
$g_{ijk} (Mem(a(i,k)),Mem(b(k,j)))$, then by Lemma~\ref{lemma:LW}
their number is then bounded by $F = \sqrt{(2M)^3}$,
so the total number of loads and stores
is bounded by
$M \lfloor G/F \rfloor = M \lfloor G/\sqrt{8M^3} \rfloor \geq
\frac{G}{\sqrt{8M}} - M$.
This proves the first lower bound:
} 

\begin{theorem}
\label{thm:Main} In the notation defined above, and in particular assuming
there are no S2/D2 arguments (created and discarded without causing memory traffic)
the number of loads and stores needed to evaluate (\ref{eqn:Model}) is at least
$\frac{G}{8\sqrt{M}} - M$.
\end{theorem}

We may also write this as $\Omega$(\#arithmetic\_operations /
$\sqrt{M}$) understanding that we only count arithmetic operations
required to evaluate the $g_{ijk}$ for $(i,j) \in S_C$ and $k \in
S_{ij}$. We note that a more careful, problem-dependent analysis
that depends on how much the three arguments can overlap, may
sometimes increase the lower bound by a factor of as much as 8,
but for simplicity we omit this.

This lower bound is not always attainable, even for dense matrix
multiplication: If the matrices are so small that they all fit in
fast memory simultaneously, so $3n^2 \leq M$, then the number of
loads and stores may be just $3n^2$, which can be much larger than
$n^3/\sqrt{M}$. So a more refined lower bound is
$\max(G/(8\sqrt{M}) - M,$\#inputs + \#outputs).
We generally omit this detail
from statements of later corollaries.

Theorem~\ref{thm:Main} is a lower bound on bandwidth, the total number of words
communicated. But it immediately provides a lower bound on latency as
well, the minimum number of messages that need to be sent, where each
message may contain many words.

\begin{corollary}
\label{cor:Latency} In the notation defined above, the number of
messages needed to evaluate (\ref{eqn:Model}) is at least
$G/(8 M^{3/2}) - 1$ = \#evaluations\_of\_$g_{ijk} /
(8M^{3/2}) - 1$.
\end{corollary}

The proof is simply that the largest possible message size
is the fast (or local) memory size $M$, so we divide
the lower bound from Theorem~1 by $M$.

On a parallel computer it is possible for a processor to pack $M$
words into a single message to be sent to a different processor.
But on a sequential computer the words to be sent in a single
message must generally be located in contiguous memory
locations, which depends on the data structures used. This
assumption is appropriate to capture the behavior of real
hardware, e.g., cache lines, memory prefetching, disk accesses,
etc. This means that to attain the latency lower bound on a
sequential computer, rather different matrix data structures may
be required than row-major or column-major
\cite{BallardDemmelHoltzSchwartz09,FrigoLeisersonProkopRamachandran99,ElmrothGustavsonJonssonKagstrom04,AndersenGustavsonWasniewski01,AhmedPingali00}.

Finally, we note that real computers typically don't have just one
level of memory hierarchy, but many,
each with its own underlying bandwidth and latency costs.
So it is of interest to minimize {\em
all} communication, between every pair of adjacent levels of the
memory hierarchy. As has been noted before
\cite{Savage95,BallardDemmelHoltzSchwartz09}, when the memory
hierarchy levels are nested (the L2 cache stores a subset of L3
cache, etc.) we can apply lower bounds like ours at every level in
the hierarchy.

\section{Consequences for BLAS, $LU$, Cholesky, and $LDL^T$}
\label{sec:Consequences}

We now show how Theorem~\ref{thm:Main} applies to a variety of
conventional algorithms from numerical linear algebra,
by which we mean algorithms that would cost $O(n^3)$ arithmetic
operations when applied to dense $n$-by-$n$ matrices, as opposed to
Strassen-like algorithms.

It is natural to ask whether algorithms exist that attain these lower
bounds. We point out cases where we know such algorithms exist, which
are therefore optimal in the sense of minimizing communication.
In the case of dense matrices, many optimal algorithms are known,
though not yet in all cases.
In the case of sparse matrices, little seems to be known.

\subsection{Matrix Multiplication and the BLAS}

We begin with matrix multiplication, on which our model in
Equation~(\ref{eqn:Model}) is based:

\begin{corollary}
\label{cor:SeqMM}
$G/(8\sqrt{M}) - M$ is the bandwidth lower bound for
multiplying explicitly stored matrices $C = A \cdot B$ on a sequential machine,
where $G$ is the number of multiplications performed in evaluating
all the $C_{ij} = \sum_k A_{ik} \cdot B_{kj}$, and $M$ is the fast memory size.
In the special case of multiplying a dense $n$-by-$r$ matrix times a dense
$r$-by-$m$ matrix, this lower bound is
$n \cdot r \cdot m/\sqrt{8M} -M$.
\end{corollary}

This nearly reproduces a result in \cite{IronyToledoTiskin04} for
the case of two {\em distinct, dense} matrices, whereas we need no
such assumptions; their bound is $\sqrt{8}$  times larger than
ours, but as stated before our bound could be improved by
specializing it to this case. We note that this result could have
been stated for {\em sparse} $A$ and $B$ in \cite{HongKung81}:
Combine their Theorem~6.1 (their $\Omega (|V|)$ is the number of
multiplications) with their Lemma~6.1 (whose proof does not
require $A$ and $B$ to be dense).

As noted in the previous section, an independent lower bound on
the bandwidth is simply the total number of inputs that need to be
read plus the number of outputs that need to be written. But
counting the number of inputs is not as simple as counting the
number of nonzero entries of $A$ and $B$: if $A$ and $B$ are
sparse, and column $i$ of $A$ is filled with zeros only, then row
$i$ of $B$ need not be loaded at all, since $C$ does not depend on
it. An algorithm that nevertheless loads row $i$ of $B$ will still
satisfy the lower bound. And an algorithm that loads and
multiplies by explicitly stored zero entries of $A$ or $B$ will
also satisfy the lower bound; this is an optimization sometimes
used in practice \cite{VuducDemmelYelick05}.

When $A$ and $B$ are dense and distinct, there are well-known
algorithms mentioned in the Introduction that (nearly) attain the
combined lower bound
\[
\Omega ( \max( n \cdot r \cdot m/ \sqrt{M}, \#inputs + \#outputs ) )
= \Omega ( \max( n \cdot r \cdot m/ \sqrt{M}, n \cdot r + r \cdot m + n \cdot m  ) ) \; \; ,
\]
see \cite{IronyToledoTiskin04} for a more
complete discussion. Attaining the corresponding latency lower
bound of Corollary~\ref{cor:Latency} requires a different data
structure than the usual row-major or column-major orders, so that
words to be sent in a single message are contiguous in memory, and
is variously referred to as {\em recursive block storage} or
storage using {\em space filling curves}, see
\cite{FrigoLeisersonProkopRamachandran99,ElmrothGustavsonJonssonKagstrom04,BallardDemmelHoltzSchwartz09}
for discussion. Some of these algorithms also minimize bandwidth
and latency for arbitrarily many levels of memory hierarchy.
Little seems to be known about the attainability of this lower
bound for general sparse matrices.

\ignore{
I don't immediately see how to determine if this is attainable;
I suspect it is not in general.
Consider the case most analogous to the dense case, where the
$n$-by-$n$ sparse matrices
$A$, $B$ and $C$ all have uniformly distributed nonzero patterns,
so that $b$-by-$b$ subblocks all have about the same number
of nonzeros, and where multiplying any two such subblocks
takes about the same number of flops. If we use the
cache-aware approach analogous to the dense case, choosing
$b$ so the number $\approx (b/n)^2 \cdot (nnz(A)+nnz(B)+nnz(C))$ of nonzeros
in any $b$-by-$b$ subblock of $A$, $B$ and $C$ is about $M$:
$b \approx n \cdot \sqrt{M / (nnz(A)+nnz(B)+nnz(C))}$, and similarly for
$B$ and $C$, and multiplying all $(n/b)^3$ pairs of
$b$-by-$b$ submatrices in fast memory, we get a number
of memory references of about $(n/b)^3 \cdot M =
(nnz(A)+nnz(B)+nnz(C))^{3/2}/\sqrt{M}$.
This is close to $T/\sqrt{M}$ if the matrices are dense,
but may not be in general. For example, consider the
case where $A$ and $B$ may be written as $S \otimes T$, where
$S$ is dense and $T$ is block diagonal.
} 

Now we consider the parallel case, with $P$ processors. Let
$nnz(A)$ be the number of nonzero entries of $A$; then $NNZ =
nnz(A) + nnz(B) + nnz(C)$ is a lower bound on the total memory
required to store the inputs and outputs. We need to make some
assumption about how this data is spread across processors (each
of which has its own memory), since if $A$, $B$ and $C$ were all
stored in one processor, and all arithmetic done there (i.e., no
parallelism at all), then no communication would be needed. So we
assume that each processor stores an equal share $NNZ/P$ of this
data, and perhaps at most $o(NNZ/P)$ more words, a kind of
memory-balance or memory-scalability assumption. Also, at least
one processor must perform at least $G/P$ multiplications, where
$G$ is the total number of multiplications. Combining all this
with Theorem~\ref{thm:Main} yields\footnote{We present the
conclusions for the parallel model in asymptotic notation. One
could instead assume that each processor had memory of size $M=\mu
\cdot \frac{n^2}{P}$ for some constant $\mu$, and obtain the
hidden constant of the lower bounds as a function of $\mu$, as
done in \cite{IronyToledoTiskin04}.}


\begin{corollary}
\label{cor:ParMM} Suppose we have a parallel algorithm on $P$
processors for multiplying matrices $C = A \cdot B$ that is
memory-balanced in the sense described above. Then at least one
processor must communicate $\Omega \lt( G/\sqrt{ P \cdot
NNZ} - NNZ/P \rt)$ words, where $G$ is the number of
multiplications $A_{ij} \cdot B_{kj}$ performed. In the special
case of dense $n$-by-$n$ matrices, this lower bound is $\Omega
\lt( n^2/\sqrt{ P} \rt)$.
\end{corollary}

There are again well-known algorithms that attain the bandwidth
and latency lower bounds in the dense case, but not in the sparse
case.

We next extend Theorem~\ref{thm:Main} beyond matrix
multiplication. The simplest extension is to the so-called BLAS3
(Level-3 Basic Linear Algebra Subroutines
\cite{BLAS,newBLASfull,newBLAS}), which include related operations
like multiplication by (conjugate) transposed matrices, by
triangular matrices and by symmetric (or Hermitian) matrices. The
last two corollaries apply to these operations without change (in
the case of $A^T \cdot A$ we use the fact that
Theorem~\ref{thm:Main} makes no assumptions about the matrices
being multiplied not overlapping).

More interesting is the BLAS3 operation TRSM,
computing $C = A^{-1} B$ where $A$ is triangular. The inner loop
of the algorithm (when $A$ is upper triangular) is
\begin{equation}
\label{eqn:TRSM}
C_{ij} = (B_{ij} - \sum_{k=i+1}^n A_{ik} \cdot C_{kj} ) / A_{ii}
\end{equation}
which can be executed in any order with respect to $j$, but only
in decreasing order with respect to $i$. None of this matters for
the lower bound, since equation~(\ref{eqn:TRSM}) still matches
Equation~(\ref{eqn:Model}), so the lower bounds apply. Sequential
algorithms that attain these bounds for dense matrices, for
arbitrarily many levels of memory hierarchy, are discussed in
\cite{BallardDemmelHoltzSchwartz09}.

We note that our lower bound also applies to the so-called
Level 2 BLAS (like matrix-vector multiplication) and Level 1 BLAS
(like dot products), but the larger lower bound \#inputs + \#outputs
is attainable.

\subsection{LU factorization}
\label{sec:LU}

Independent of sparsity and
pivot order, the formulas describing LU factorization are
as follows, with the understanding the summations may be over
some subset of the indices $k$ in the sparse case, and
pivoting has already been incorporated in the interpretation
of the indices $i$, $j$ and $k$.
\begin{eqnarray}
\label{eqn:LUModel}
L_{ij}  &=&  (A_{ij} - \sum_{k < j} L_{ik} \cdot U_{kj})/U_{jj}
\; \; {\rm for} \; i>j \\
U_{ij}  &=&  A_{ij} - \sum_{k<i} L_{ik} \cdot U_{kj}
\; \; {\rm for} \; i \leq j \nonumber
\end{eqnarray}
It is easy to see that these formulas correspond to our model in
Equation~(\ref{eqn:Model}), with $g_{ijk}$ identified with
multiplying $L_{ik} \cdot U_{kj}$. The fact that the ``outputs''
$L_{ij}$ and $U_{ij}$ can overwrite the inputs does not matter,
and the subtraction from $A_{ij}$ and division by $U_{jj}$ are all
accommodated by Equation~(\ref{eqn:Model}).

With this in mind, these formulas are also general enough to
accommodate incomplete LU (ILU) factorization \cite{Saad96} where
some entries of $L$ and $U$ are omitted in order to speedup the
computation. If we model an ILU implementation that is {\em
threshold based}, i.e., one that computes a possible nonzero entry
$L_{ij}$ or $U_{ij}$ and compares it to a threshold, storing it
only if it is larger than the threshold and discarding it
otherwise, then we should not count the multiplications that led
to the discarded output. Thus we see that analogs of
Corollaries~\ref{cor:SeqMM} and \ref{cor:ParMM} apply to LU and
ILU as well.

\ignore{
\begin{corollary}
\label{cor:SeqLU}
$G/\sqrt{8M} - M$ is the bandwidth lower bound for computing
the (complete or incomplete) LU factorization of a matrix on a sequential machine,
where $G$ is the
number of multiplications performed when multiplying $L_{ik} \cdot U_{kj}$,
and $M$ is the fast memory size.
In the special case of the complete factorization of a dense $n$-by-$n$
matrix, this lower bound is $n^3/\sqrt{72M} - n^2/\sqrt{32M} + n/\sqrt{288M} - M$.
\end{corollary}
}


\ignore{
\begin{corollary}
\label{cor:ParLU}
Suppose we have a parallel algorithm on $P$
processors for computing the complete or (non-threshold-based) incomplete LU factorization.
Let $NNZ = nnz(L) + nnz(U)$ be the size of the output where we assume
each processor stores an equal share of this data $NNZ/P$.
Then at least one processor must communicate at least
$(G/P)/\sqrt{8 \cdot NNZ/P} - NNZ/P = G/\sqrt{8 \cdot P \cdot NNZ} - NNZ/P$
words, where $G$ is the
number of multiplications performed when multiplying $L_{ik} \cdot U_{kj}$.
In the special case of the complete factorization of a dense $n$-by-$n$
matrix, this lower bound is $n^2/\sqrt{72P} - n/ \sqrt{32P} + 1/\sqrt{288P} - n^2/P$.
\end{corollary}
}

A sequential dense LU algorithm that attains this bandwidth lower bound is
given by \cite{Toledo97}, although it does not always attain the
latency lower bound \cite{DemmelGrigoriHoemmenLangou08a}.
The conventional parallel dense LU algorithm implemented in
ScaLAPACK \cite{SCALAPACK} attains the bandwidth lower bound
(modulo an $O(\log P)$ factor), but not the latency lower bound.
A parallel algorithm that attains both lower bounds (again modulo a factor $O(\log P)$)
is given in \cite{DemmelGrigoriXiang08}, where significant speedups are reported.
Interestingly, it does not seem possible
to attain both lower bounds and retain conventional partial pivoting,
a different (but still stable) kind of pivoting is required.
We also know of no dense sequential LU algorithm that minimizes
bandwidth and latency across multiple levels of a memory hierarchy
(unlike Cholesky).
There is an elementary reduction proof that dense LU factorization
is ``as hard as dense matrix multiplication'' \cite{DemmelGrigoriHoemmenLangou08a},
but it does not address sparse or incomplete LU, as does our approach.

\subsection{How to carefully count operations $g_{ijk}$
}\label{sec:careful}

Using only structural information, e.g., about the sparsity
patterns of the underlying matrices, it is sometimes possible to
deduce that the computed result $f_{ij}( \cdot )$ is exactly zero,
and so to possibly avoid a memory reference to location $c(i,j)$
to store the result. This may either be because the values
$g_{ijk}( \cdot )$ being accumulated to compute $f_{ij}$ are all
identically zero, or, more interestingly, because it is possible
to prove there is exact cancellation (independent of the values of
the nonzero arguments $Mem(a(i,k))$ and $Mem(b(k,j))$); we give
an example of this below.
In these cases it is possible to imagine algorithms that (1) pay
no attention to these possibilities and simply load, store and
compute with zeros (e.g., a ``dense algorithm'' applied to a
sparse matrix), or (2) recognize that zeros will be computed and
avoid doing any memory traffic or arithmetic to compute them, or
(3) do the work of computing the zero entry, recognize it is zero
(perhaps by comparing to a tiny threshold), and do not bother
storing it. In this last case one might worry that our lower
bounds are too high. However, such operations would not count
toward our lower bound, because they do not satisfy
Equation~(\ref{eqn:Model}), since they do not lead to a write to
$Mem(c(i,j))$. This may undercount the actual number of memory
operations, but does not prevent our lower bound from being a
right lower bound.

Here is an example to illustrate that counting the number of
$g_{ijk}$ to get a true lower bound requires care. This is
because, as suggested above, it is
possible for an LU algorithm to infer from the sparsity pattern of
$A$ that some (partial) sums in equation (\ref{eqn:LUModel}) are
zero and so avoid computing them. For example, consider a matrix
$A$ that is nonzero in its first $r$ rows and columns, and
possibly in the trailing $(n-2r)$-by-$(n-2r)$ submatrix; call this
submatrix $A'$. First suppose $A'=0$, so that $A$ has rank at most
$2r$, and that pivots are chosen along the diagonal. It is easy to
see that the first $2r-1$ steps of Gaussian elimination will
generically fill in the entire matrix with nonzeros, but that step
$2r$ will cause cancellation to zero (in exact arithmetic) in all
entries of $A'$. If $A'$ starts as a nonzero sparse matrix, then
this cancellation will not be to zero but to the sparse LU
factorization of $A'$ alone. So one can imagine an algorithm that
may or may not recognize this opportunity to avoid work in some or
all of the entries of $A'$. To accommodate all these
possibilities, we will, as stated above,
only count those multiplications in
(\ref{eqn:LUModel}) that contribute to a result $L_{ij}$ or
$U_{ij}$ that is stored in memory. The discussion of this
paragraph also applies to QR factorization.

\subsection{Cholesky Factorization}

Now we consider Cholesky factorization. Independent of sparsity and
(diagonal) pivot order, the formulas describing Cholesky factorization are
as follows, with the understanding the summations may be over
some subset of the indices $k$ in the sparse case, and
pivoting has already been incorporated in the interpretation
of the indices $i$, $j$ and $k$.
\begin{eqnarray}
L_{jj} & = & (A_{jj} - \sum_{k < j} L_{jk}^2)^{1/2} \\
L_{ij} & = & (A_{ij} - \sum_{k<j} L_{ik} \cdot L_{jk})/L_{jj}
\; \; {\rm for} \; i > j \nonumber
\end{eqnarray}
It is easy to see that these formulas correspond to our model in
Equation~(\ref{eqn:Model}), with $g_{ijk}$ identified with
multiplying $L_{ik} \cdot L_{jk}$. As before, the fact that the
``outputs'' $L_{ij}$ can overwrite the inputs does not matter, and
the subtraction from $A_{ij}$, division by $L_{ii}$, and square
root are all accommodated by Equation~(\ref{eqn:Model}). As
before, these formulas are general enough to accommodate
incomplete Cholesky (IC) factorization \cite{Saad96}. \looseness
-1000

Dense algorithms that attain these lower bounds are discussed in
\cite{BallardDemmelHoltzSchwartz09}, both parallel and sequential,
including analyzing one that minimizes bandwidth and latency
across all levels of a memory hierarchy
\cite{AhmedPingali00}.
We note that there was a proof in
\cite{BallardDemmelHoltzSchwartz09} showing that dense Cholesky
was ``as hard as dense matrix multiplication'' by a method
analogous to that for LU.

\subsubsection{Sparse Cholesky Factorization on Matrices whose Graphs are Meshes}

Hoffman, Martin, and Rose \cite{HMR73} and George \cite{George73}
prove that a lower bound on the number of multiplications required
to compute the sparse Cholesky factorization of an $n^2$-by-$n^2$
matrix representing a 5-point stencil on a 2D grid of $n^2$ nodes
is $\Omega(n^3)$. This lower bound applies to any matrix
containing the structure of the 5-point stencil. This yields:

\begin{corollary}
\label{cor:sparseChol} In the case of the sparse Cholesky
factorization of the matrix representing a
5-point stencil on a two-dimensional grid of $n^2$ nodes, the bandwidth
lower bound is $\Omega\lt(\frac{n^3}{\sqrt M} \rt)$.
\end{corollary}

George \cite{George73} shows that this arithmetic lower bound is
attainable with a nested dissection algorithm in the case of the
5-point stencil.
Gilbert and Tarjan \cite{GilbertTarjan87} show that
the upper bound also applies to a larger class of structured
matrices, including matrices associated with planar graphs.


\subsection{$LDL^T$ Factorization}
We next show that analogous lower bounds apply to the symmetric
indefinite factorization \linebreak $A = LDL^T$, where $D$ is
block diagonal with 1-by-1 and 2-by-2 blocks, and $L$ is a lower
triangular matrix with 1 on its diagonal elements. If $A$ is
positive definite then all the blocks of $D$ are 1-by-1. It
sufficient to consider the lower bound for the positive definite
case, as any $LDL^T$ decomposition algorithm for the general case
has to deal with this case as well. Independent of sparsity and
(diagonal) pivot order, the formulas describing Cholesky
factorization are as follows, with the understanding the
summations may be over some subset of the indices $k$ in the
sparse case, and pivoting has already been incorporated in the
interpretation of the indices $i$, $j$ and $k$.
\begin{eqnarray}
D_{jj} & = & A_{jj} - \sum_{k < j} L_{jk}^2 D_k \\
L_{ij} & = & \frac{1}{D_{jj}}\lt(A_{ij} - \sum_{k<j} L_{ik} \cdot
L_{jk}D_{kk}\rt) \; \; {\rm for} \; i > j
\end{eqnarray}
It is easy to see that these formulas correspond to our model in
Equation~(\ref{eqn:Model}), with $g_{ijk}$ identified with
multiplying $L_{ik} \cdot L_{jk}$, thus the lower bounds apply to
$LDL^T$ decomposition as well. As in the Cholesky case, the only
difference is that the number of multiplications $G$ is about half
as large ($n^3/6 + O(n^2)$), as is the memory size $NNZ =
n(n+1)/2$.

\section{Applying Orthogonal Transformations}
\label{sec:QR}

In this section we consider algorithms that compute and apply
sequences of orthogonal transformations to a matrix, which
includes the most widely used algorithms for least squares
problems (the QR factorization), eigenvalue problems, and the SVD.
We need to treat algorithms that apply orthogonal transformations
separately because Loomis-Whitney alone is not enough to bound the
number of arithmetic operations that can occur in a segment.

\subsection{The QR Factorization}
The QR factorization of a rectangular matrix $A$ is more subtle to
analyze than LU or Cholesky, because there is more than one way to
represent the $Q$ factor (e.g., Householder reflections and Givens
rotations), because the standard ways to reorganize or ``block''
QR to minimize communication involve using the distributive law,
not just summing terms in a different order
\cite{BischofVanLoan87,Schreiber87a,Puglisi92,Demmel97,GolubVanloan96},
and because there may be many intermediate terms that are
computed, used, and discarded without causing any slow memory
traffic. This forces us to use a different argument than
Loomis-Whitney to bound the number of arithmetic operations in a
segment.

To be concrete, we consider the widely used Householder reflections,
in which an $n$-by-$n$ elementary real orthogonal matrix
$Q_i$ is represented as $Q_i = I - \tau_i u_i u_i^T$,
where $u_i$ is a column vector called a Householder vector,
and $\tau_i = 2/\|u_i\|_2^2$.
A single Householder reflection $Q_i$ is chosen so that
multiplying $Q_i \cdot A$
zeros out selected rows in a particular column of $A$,
and modifies one other row in the same column
(for later use, we let $r_i$ be index of this other row).

Once entries in a column are zeroed out, they are not operated on
again, and remain zero; this fact will be critical to our later
counting argument. This means that a sequence of $k$ Householder
reflections is chosen to zero out a common set of selected rows of
a sequence of $k$ columns of $A$ by forming $Q_k \cdot Q_{k-1}
\cdots Q_1 \cdot A$. Each $Q_i$ zeros out a subset of the rows
zeroed out in the previous column (so that once created, zeros are
preserved). We furthermore model the way libraries like LAPACK
\cite{LAPACK} and ScaLAPACK \cite{SCALAPACK} may ``block''
Householder vectors, writing $Q_k \cdots Q_1 = I - U_k T_k U_k^T$,
where $U_k = [u_1, u_2 , \ldots , u_k]$ is $n$-by-$k$, and $T_k$
is $k$-by-$k$. $U_k$ is nonzero only in the rows being modified,
and furthermore column $i$ of $U_k$ is zero in entries
$r_1$,...,$r_{i-1}$ and nonzero in entry $r_i$. (In conventional
algorithms for dense matrices this means $r_i=i$, and $U_k$ is
lower trapezoidal with nonzero diagonal.)
Furthermore $T_i$, which may be computed recursively from $T_{i-1}$,
is lower triangular with nonzero diagonal.
Our lower bound considers all possible sequences of Householder
transformations that preserve previously created zeros,
and all possible ways to collect them into blocks.

\ignore{ We must also assume the algorithm does no ``useless''
arithmetic involving multiplying by or adding numbers it knows to
be zero for structural reasons (i.e., the sparsity pattern of
$A$), so that they require no memory traffic to get; otherwise the
lower bound, which is proportional to \#flops, could be
artificially large. {\em This probably needs to be assumed
elsewhere as well.} This does not preclude storing and doing
arithmetic on explicitly zeros, i.e., not using any knowledge that
they are zero (this is done in some optimizations of sparse matrix
algorithms, to simplify indexing \cite{OSKI}).
} 

Next, we will apply such block Householder
transformations to a (sub)matrix by inserting parentheses as follows:
$(I - U \cdot T \cdot U^T) \cdot A = A - U \cdot (T \cdot U^T \cdot A) \equiv A - U \cdot Z$,
which is also the way Sca/LAPACK does it.
Finally, we overwrite the output onto $A = A - U \cdot Z$, which
is how all fast implementations do it, analogously to LU decomposition,
to minimize memory requirements.

But we do not need to assume any more commonality with the
approach in Sca/LAPACK, in which a vector $u_i$ is chosen to zero
out all of column $i$ of $A$ below the diagonal. For example, we
can choose each Householder vector to zero out only part of a
column at a time, as is the case with the algorithms for dense
matrices in
\cite{DemmelGrigoriHoemmenLangou08a,DemmelGrigoriHoemmenLangou08b}.
Nor do we even need to assume we are zeroing out any particular
set of entries, such as those below the main diagonal as the usual
QR algorithm; later this generality will let us apply this result
to algorithms for eigenproblems and the SVD. As stated before, all
we assume is that the algorithm ``makes progress'' in the sense
that a Householder transformation in column $j$ is not allowed to
fill in zeros that were deliberately created by other Householder
transformations in columns $k<j$.

To get our lower bound, we consider just the multiplications in
all the different applications of block Householder
transformations $A = A - U \cdot Z$. We argue in Section
\ref{sec:QR-discussion} that this constitutes a large fraction of
all the multiplications in the algorithm. 
There are two challenges to straightforwardly applying our
previous approach to the matrix multiplications in all the updates
$A = A - U \cdot Z$. The first challenge is that we need to
collect all these multiplications into a single set indexed in an
appropriate one-to-one fashion by $(i,j,k)$. The second challenge
is that $Z$ need not be read from memory, rather it may be
computed on-the-fly from $U$ and $A$ and discarded without
necessarily ever being read or written from memory. So we have to
account for its memory traffic more carefully. Furthermore, each
Householder vector (column of $U$) is created on the fly by
modifying certain (sub)columns of $A$, so it is both an output and
an input. So we will have to account for the memory operations for
$U$ and $Z$ more carefully.

Here is how we address the first challenge: Let index $k$ indicate
the number of the  Householder vector; in other words $U(:,k)$ are
all the entries of $k$-th Householder vector (the ordering is
arbitrary). Thus $k$ is not the column of $A$ from which $U(:,k)$
arises (there may be many Householder vectors associated with a
given column as in \cite{DemmelGrigoriHoemmenLangou08a}) but $k$
does uniquely identify that column. Then the operation $A - U
\cdot Z$ may be rewritten as $A(i,j) - \sum_k U(i,k)\cdot Z(k,j)$,
where the sum is over the Householder vectors $k$ making up $U$
that both lie in column $j$ and have entries in row $i$. The use
of this index $k$ lets us combine all the operations $A = A - U
\cdot Z$ for all different Householder vectors into one collection
\begin{equation}
\label{eqn:UZmodel}
A(i,j) = A(i,j) - \sum_k U(i,k) \cdot Z(k,j)
\end{equation}
where all operands $U(i,k)$ and $Z(k,j)$ are uniquely labeled by
the index pairs $(i,k)$ and $(k,j)$, resp.

For the second challenge, we revisit the model of
section~\ref{sec:MainResult}, where we distinguished
arguments by their sources (S1 or S2) and destinations
(D1 or D2). Unlike the algorithms in section~\ref{sec:Consequences},
we now have the possibility of S2/D2 arguments, $Z(k,j)$, which
are created during the segment and then discarded without causing
any slow memory traffic. To bound the number of S2/D2 arguments,
we need to exploit the mathematical structure of Householder
transformations.


Now let us consider the possible labels Si/Dj for the arguments
in model~(\ref{eqn:UZmodel}). We assume each
required value is only computed once (reflecting
practical implementations).
\begin{description}
\item[$A(i,j)$:] Every $A(i,j)$ operand is destined either to be output
(eg as an entry of the $R$ factor) or converted into a Householder
vector. So the only possible S2/D2 operands from $A$ are (sub)columns
that become Householder vectors, and hence become S2 operands of $U$.
We bound the number of these as follows.
\item[$U(i,k)$:] All $U$ operands are eventually output, so
there are no D2 operands of $U$ (recall that we may only compute
each result $U(i,k)$ once, so it cannot be discarded). So all S2
operands $U(i,k)$ are also D1, and so there are at most $2M$ of
them. This also bounds the number of S2/D2 operands $A(i,j)$, and
so bounds the total number of $A(i,j)$ operands by $4M$.
\item[$Z(k,j)$:] There are possible S2/D2 operands, as many as the
product of the number of columns $U(:,k)$ and the number of
columns $A(:,j)$ that can reside entirely in fast memory during
the segment.
\end{description}

Unlike LU and matrix multiplication, the number of $Z(k,j)$
arguments is {\em not} bounded by $4M$, so we cannot hope to apply
Loomis-Whitney to bound the number of arithmetic operations, see
Section~\ref{sec:QR-discussion} for an example. So our proof
has to try to bound the number of arithmetic operations possible
in a segment without using Loomis-Whitney. We do this by
exploiting the mathematical structure of Householder
transformations, and the $O(M)$ bound on the number of $A(i,j)$
and $U(i,k)$ entries in fast memory during a segment, to still
bound the maximum number of multiplies during a segment by
$O(M^{3/2})$, which is all we need to get our ultimate lower
bound.

In the case that all $Z(k,j)$ are S2/D2 operands, according to (\ref{eqn:UZmodel}),
the maximum number of arithmetic operations
for column $j$ of $A$ is the number of $U(i,k)$ entries, $2M$,
if every possible $Z(k,j)$ is nonzero. So the maximum total
number of arithmetic operations is $2M$ times the maximum number
of columns of $A$ that can be present. The maximum number of
columns of $A$ is in turn the total number of $A(i,j)$ entries
that can be present ($4M$) divided by the minimum number of
rows present from each column of $A$. So the question is: what is
the fewest number of rows in which we can pack $2M$ Householder
vector entries? We need to rely on the fact that these Householder
vectors must satisfy the dependency relationship of QR, that previously
created zeros in previous columns are preserved. So the fewest rows
are touched when the nonzero Householder vector entries in each
column lie in a strict subset of the rows occupied by nonzero Householder vector
entries in previous columns (this follows by induction on columns
from right to left). Note that if the matrix is sparse, these
Householder vectors are not necessarily in adjacent columns.
So if there are $r$ nonzero Householder vector
entries in the leftmost occupied column, there can be at most $r-1$
in the next occupied column, and so on,
and so at most $r(r+1)/2$ nonzero Householder
vector entries altogether in $r$ rows. So if there are $0 < h \leq 2M$
nonzero Householder vector entries altogether, residing in $r$ rows,
then $r(r+1)/2 \geq h$ or $r \geq 2 h^{1/2} -1 \geq h^{1/2}$.
Then the maximum number $c$ of columns of $A$ that can be present is
bounded by $r \cdot c \leq 4M$ or
$c \leq  4M/r \leq 4M/h^{1/2}$, and the maximum number of
multiplications that can be performed is
$h \cdot c \leq  h \cdot 4M/h^{1/2} = 4Mh^{1/2} \leq \sqrt{32M^3}$,
as desired.

In the case that some $Z(k,j)$ are S2/D2 operands and some are not, we
must use Loomis-Whitney and the above argument to bound the number
of multiplies within a segment.  Since there are no more that $4M$ non-S2/D2
$Z(k,j)$ operands in a segment, the Loomis-Whitney argument bounds the
number of multiplies involving such operands by $(32M^3)^{1/2}$, so with
the above argument, the total number of multiplies is less than $\sqrt{128M^3}$.

The rest of the proof is similar to before: A lower bound on the
number of segments is then $\lfloor {\rm \#multiplies}
/(128M^3)^{1/2} \rfloor \geq {\rm \#multiplies} /(128M^3)^{1/2}
-1$, so a lower bound on the number of slow memory accesses is $M
\cdot \lfloor {\rm \#multiplies} /(128M^3)^{1/2} \rfloor \geq {\rm
\#multiplies} /(128M)^{1/2} -M$. For dense $m$-by-$n$ matrices
with $m \geq n$, the conventional algorithm does $\Theta (mn^2)$
multiplies.
Altogether, this yields the following theorem:
\begin{theorem}
\label{thm:QR} $G/\sqrt{128M} - M$ is the bandwidth lower bound
for computing the QR factorization of a matrix on a sequential
machine, where $Q$ is computed as a product of (arbitrarily
blocked) Householder transformations, and $G$ is the number of
multiplications performed when updating $A(i,j) = A(i,j) + U(i,k)
\cdot Z(k,j)$, i.e., adding multiples of Householder vectors to
the matrix. In the special case of a dense $m$-by-$n$ matrix with
$m \geq n$, this lower bound is $\Omega (mn^2 / \sqrt{M})$.
\end{theorem}

An analogous result holds for parallel QR.

Regarding related work, a lower bound just for the dense case
appears in \cite{DemmelGrigoriHoemmenLangou08a}, which also
discusses algorithms that meet (some of) these lower bounds, and
shows significant speedups. The proof in
\cite{DemmelGrigoriHoemmenLangou08a} assumed each $Z$ was written
to memory (which we do not here), and so could use the
Loomis-Whitney approach. For example, the sequential algorithm in
\cite{ElmrothGustavson98,Elmroth00}, as well as LAPACK's DGEQRF,
can minimize bandwidth, but not latency. The paper
\cite{DemmelGrigoriHoemmenLangou08a} also describes sequential and
parallel QR algorithms that do attain both bandwidth and latency
lower bounds; this was accomplished by applying block Householder
transformations in a tree-like pattern, each one of which zeroing
out only part of a block column of $A$. We know of no sequential
QR factorization algorithm that minimizes communication across
more than 2 levels of the memory hierarchy.

\ignore{
{\em
Questions: Do you agree with this? Does Loomis-Whitney go
away entirely, or should we still use it for the non-S2/D2 operands?
Should we try to account for the $Z = U^T \cdot A$ operations in the
same analysis, or resort to our previous argument that $Z = U^T \cdot A$
takes no more flops than $U \cdot Z$?

To do eigenvalue problems, e.g., reduction to tridiagonal form, we
need to consider two-sided transformations, where we apply
transformation to the left and the right of the matrix. In the
dense case at least, it is tempting to just ignore the left
transformations, which do as much work as the right, and just
apply our lower bound to the transformations on the right, but
since we are relying on certain properties of one-sided QR, we'd
have to be sure that the alternating transformations from the
other side don't muck up the argument, since we are not just using
Loomis-Whitney anymore.

Finally, does this argument work for LU?
}
} 


\subsubsection{Discussion of QR Model}\label{sec:QR-discussion}

To get our lower bound, we consider just the multiplications in
all the different applications of block Householder
transformations $A = A - U \cdot Z$, where $Z=T\cdot U^T\cdot A$.
We argue that under a natural ``genericity assumption" this
constitutes a large fraction of all the multiplications in the
algorithm. (although this is not necessary to get a valid lower
bound). Suppose $(U^T \cdot A)(k,j)$ is nonzero; the amount of
work to compute this is at most proportional to the total number
of entries stored (and so treated as nonzeros) in column $k$ of
$U$. Since $T$ is triangular and nonsingular, this means $Z(k,j)$
will be generically nonzero as well, and will be multiplied by
column $k$ of $U$ and added to column $j$ of $A$, which costs at
least as much as computing $(U^T \cdot A)(k,j)$. The cost of the
rest of the computation, forming and multiplying by $T$ and
computing the actual Householder vectors, are lower order terms in
practice; the dimension of $T$ is chosen small enough by the
algorithm to try to assure this.  Thus, for a lower bound of
$\Omega (mn^2)$ total multiplies for a dense $m$-by-$n$ matrix
given by \cite{DemmelGrigoriHoemmenLangou08a}, there are
$\Omega(mn^2)$ multiplies of the form given by
(\ref{eqn:UZmodel}).

Unlike LU and matrix multiplication, the number of $Z(k,j)$
arguments is {\em not} bounded by $2M$, so we cannot hope to apply
Loomis-Whitney to bound the number of arithmetic operations. As an
example, suppose we do QR on the matrix $[A_1,A_2]$ where each
$A_i$ is $n$-by-$n$, and the matrix just fits in fast memory, so
that $2n^2 \approx M$. Suppose that we have performed QR on $A_1$
using $n(n-1)/2$ 2-by-2 Householder transformations (for example
zeroing the subdiagonal entries of $A_1$ from column 1 to column
$n-1$, and from bottom to top in each column). Now suppose we want
to bound the number of the operations (\ref{eqn:UZmodel}) gotten
from applying these Householder transformations to $A_2$. Then
there is one (generically) nonzero $Z(k,j)$ for each pair
consisting of a Householder vector $k=1,...,n(n-1)/2$ and a column
$j=1,...,n$ of $A_2$, or $n^2(n-1)/2 = \Theta (n^3) = \Theta
(M^{3/2})$ values of $Z(k,j)$ in all. This is too large to use
Loomis-Whitney to bound the number of multiplications in the
(single) segment constituting the algorithm, which is still
$\Theta (n^3) = \Theta (M^{3/2})$.

We note that we cannot use the fact that $R^T \cdot R = A^T \cdot
A$, so that $R$ is the Cholesky factor of $A^T \cdot A$, to get a
lower bound based on our lower bound for Cholesky, because QR
works very differently than Cholesky. For example, when $A$ is
$m$-by-$n$ with $m \ll n$, its $R$ factor will have the same shape
and cost $O(mn^2)$ operations to compute, but $A^T \cdot A$ will
be $n$-by-$n$, as will its Cholesky factor, which would cost
$O(n^3)$ to compute. Similarly, when $A$ is sparse, $A^T \cdot A$
may be much denser. So there is no simple relationship between the
two algorithms.

We note that a Givens rotation may be replaced by a 2-by-2
Householder reflection, and so we conjecture that algorithms using
Givens rotations are covered by this analysis.

\subsection{Eigenvalue and Singular Value Problems}

Standard algorithms for computing eigenvalues and eigenvectors,
or singular values and singular vectors (the SVD), start by applying
orthogonal transformations to both sides of $A$ to reduce
it to a ``condensed form'' (Hessenberg, tridiagonal or bidiagonal)
with the same eigenvalues or singular values, and simply related
eigenvectors or singular vectors \cite{Demmel97}.
We begin this section by getting communication lower bounds for
these reductions, and then discuss the communication complexity
of the algorithms for the condensed forms.
Finally, we briefly discuss a completely different family of
algorithms that does attain the same lower bounds for all these
eigenvalue problems and the SVD, but at the cost of doing more
arithmetic.

We extend our argument from the last section as follows.
We can have some arbitrary interleaving of (block) Householder transformations
applied on the left:
\[
A = (I - U_L \cdot T_L \cdot U_L^T) \cdot A
  = A - U_L \cdot (T_L \cdot U_L^T \cdot A)
  \equiv A - U_L \cdot Z_L
\]
and the right:
\[
A = A \cdot (I - U_R \cdot T_R \cdot U_R^T)
  = A - ( A \cdot U_R \cdot T_R) \cdot U_R^T
  \equiv A - Z_R \cdot U_R^T \; \; .
\]
Combining these, we can write
\begin{equation}
\label{eqn:TwoSidedModel}
A(i,j) = A(i,j) - \sum_{k_L} U_L(i,k_L ) \cdot Z_L(k_L ,j)
                - \sum_{k_R} Z_R(i,k_R ) \cdot U_R(j,k_R )
\end{equation}
Of course there are lots of possible dependencies ignored
here, much as we wrote down a similar formula for LU.
But we do assume that two-sided factorizations ``make progress''
in the same sense as one-sided QR: entries that are zeroed out
remain zeroed out by subsequent Householder transformations,
from the left or right.
By the same argument as the previous section, we can classify the Sources
and Destinations of the components of (\ref{eqn:TwoSidedModel})
as follows:

\begin{description}
\item[$A(i,j)$:] Every $A(i,j)$ operand is destined either to be output
(e.g., as an entry of the eventual Hessenberg, tridiagonal or bidiagonal matrix)
or converted into a left or right Householder vector.
So the only possible S2/D2 operands from $A$ are (sub)columns or (sub)rows
that become Householder vectors, and hence become S2 operands of either $U_L$ or $U_R$.
We bound the number of these as follows.
\item[$U_L(i,k_L)$:] All $U_L$ operands are eventually output, so
there are no D2 operands of $U_L$ (recall that we may only compute
each result $U_L(i,k_L)$ once, so it cannot be discarded). So all
S2 operands $U_L(i,k_L)$ are also D1, and so there are at most $2M$ of
them. This also bounds the number of S2/D2 operands $A(i,j)$ that become
S2 operands of $U_L$.
\item[$U_R(j,k_R)$:] This is analogous to $U_L(i,k_L)$. Again,
$2M$ bounds the number of S2/D2 operands $A(i,j)$ that become S2
operands of $U_R$.
\end{description}
Finally, since S2/D2 operands of $A(i,j)$ must either be S2 operands of
$U_L$ or $U_R$, there are at most $4M$ of these.

\ignore{
Thus, in one segment, there can
be at most $6M$ entries $A(i,j)$, $2M$ entries
$U_L(i,k_L )$ and $2M$ entries $U_R(j,k_R)$. The same
argument as above (using Loomis-Whitney for non-S2/D2 operands)
says that we can do at most
$(288M^3)^{1/2}$ multiplications
$U_L(i,k_L ) \cdot Z_L(k_L ,j)$ and also
$(288M^3)^{1/2}$ multiplications
$Z_R(i,k_R ) \cdot U_R(j,k_R )$, or $(1152M^3)^{1/2}$ in all.
The final lower bound on the number of memory operations is then
${\rm \#multiplies}/(1152M)^{1/2} - M$.
} 

Thus, in one segment, there can
be at most $6M$ entries $A(i,j)$, $2M$ entries
$U_L(i,k_L )$ and $2M$ entries $U_R(j,k_R)$.  Since there are no more than $4M$ non-S2/D2 $Z_L(k_L,j)$ operands and $4M$ non-S2/D2 $Z_R(k_R,j)$ operands in a segment, the Loomis-Whitney argument bounds the number of multiplies $U_L(i,k_L)\cdot Z_L(k_L,j)$ or $U_R(i,k_R)\cdot Z_R(k_R,j)$ involving such operands by $(48M^3)^{1/2}+(48M^3)^{1/2}=(192M^3)^{1/2}$.  An argument similar to the one given in section \ref{sec:QR} bounds the number of multiplies involving S2/D2 operands by $(72M^3)^{1/2}+(72M^3)^{1/2}=(288M^3)^{1/2}$.  Thus, the upper bound on the total number of multiplies within a segment is $(192M^3)^{1/2}+(288M^3)^{1/2}<(1152M^3)^{1/2}$, so the final lower bound on the number of memory operations is \#multiplies/$(1152M)^{1/2} - M$.

This extends our lower bound to any algorithm
that applies any sequence of left/right Householder transformations,
under the restriction of ``making progress'', and so cover
reduction to tridiagonal, bidiagonal or Hessenberg forms.
In all these cases, for dense $n$-by-$n$ matrices,
\#multiplies is a multiple of $n^3$.

\begin{theorem}
\label{thm:2SidedReduction}
$G/\sqrt{1152M} - M$ is the bandwidth lower bound for reducing a
matrix to Hessenberg, tridiagonal or bidiagonal form on a sequential machine,
where the reduction is done by multiplying on the left and right by products of
(arbitrarily blocked) Householder transformations,
and $G$ is the number of multiplications
performed when
adding multiples of Householder vectors to the matrix.
In the special case of a dense $n$-by-$n$ matrix,
this lower bound is $\Omega (n^3 / \sqrt{M})$.
\end{theorem}

Again, an analogous result holds for parallel reductions.

None of the reduction algorithms in LAPACK \cite{LAPACK}
attain these bounds, instead having bandwidth $O(n^3)$
(the worst possible, asymptotically, even though these
algorithm try to do as much work with matrix-matrix multiplication
as possible). ScaLAPACK's tridiagonal and bidiagonal reduction
routines minimize bandwidth but not latency
\cite[Table 5.8]{SCALAPACK}.

Now we consider the rest of the eigenvalue or singular value problem.
Once a (symmetric) matrix has been reduced to tridiagonal form $T$,
it of course requires much less memory to store, just $O(n)$.
Assuming $M$ is at least a few times larger than $n$,
there are a variety of classical algorithms to compute some or all of $T$'s
eigenvalues also using just $O(n)$ fast memory.
There is also a well-known algorithm \cite{DhillonParlettVoemel06,DhillonParlett04}
(routine xSTEMR in LAPACK \cite{LAPACK})
to compute $T$'s eigenvalues and eigenvectors one-at-a-time using
$O(n^2)$ flops, and requiring $O(n)$ fast memory in general.
So in the common case that $n$ is at least a few times smaller than
the fast memory size $M$, this can be done with as many
slow memory references as there are inputs and outputs, which is
a lower bound. A similar discussion applies to the SVD of a
bidiagonal matrix $B$, although there are open numerical stability problems
regarding extending the algorithm in xSTEMR to the SVD.
Once the eigenvectors of $T$ or singular vectors of $B$ have been
computed, they must be multiplied by the orthogonal matrices used
in the reduction to get the final eigenvectors or singular vectors
of $A$. Our previous analysis of applying Householder transformations
applies here.

Now we consider the more challenging computation of the eigenvalues
and eigenvectors of a Hessenberg matrix $H$.
Our analysis applies to one pass of standard QR iteration on a
dense upper Hessenberg matrix to find its eigenvalues,
but this does $O(n^2)$ flops on $O(n^2)$ data,
and so does not improve the trivial lower bound  of the input size.
We conjecture that
improvements of Braman, Byers and Mathias \cite{BramanByersMathias02a,BramanByersMathias02b}
to combine $m$ passes into one
increase the flop count to $O(mn^2)$, so we get a lower
bound of $\Omega (mn^2 / M^{1/2})$. This starts to get interesting
as soon as $m > M^{1/2}$. In practice, for numerical reasons, $m$
is usually chosen to be 256 or lower,
which limits the applicability of this result.

We conjecture that our analysis applies to solving the generalized
eigenvalues problem of a matrix pencil $A - \lambda B$:
The standard algorithm begins by
reducing the pair $(A,B)$ to Hessenberg/triangular form
by applying orthogonal transformations to the left and right of $A$ and $B$.
After this the QZ algorithm is used to find eigenvalues and eigenvectors.

Finally, there is a completely different, divide-and-conquer
approach to solving dense eigenproblems and the SVD
\cite{DemmelDumitriuHoltz07,BallardDemmelDumitriu09}, that only
uses QR factorization and matrix multiplication to do its work,
and attains the communication lower bounds described above. We
discuss this briefly in Section \ref{sec:open}.


\ignore{ Regarding related work, we note that a communication
lower bound just for the dense case appears in
\cite{DemmelGrigoriHoemmenLangou08a}, which also has an extensive
discussion of algorithms that meet (some of) these lower bounds.
(This proof (unnecessarily) assumed each $Z$ was written to memory
(which we do not here), and so could use Loomis-Whitney approach.)
For example, the sequential algorithm in
\cite{ElmrothGustavson98,Elmroth00}, as well as LAPACK's DGEQRF,
can minimize bandwidth, but not latency. The paper
\cite{DemmelGrigoriHoemmenLangou08a} also describes sequential and
parallel QR algorithms that do attain both bandwidth and latency
lower bounds; this was accomplished by applying block Householder
transformations in a tree-like pattern, each one of which zeroing
out only part of a block column of $A$. We know of no sequential
QR factorization algorithm that minimizes communication across
more than 2 levels of the memory hierarchy.
} 

\section{Lower Bounds for Compositions of Linear Algebra
Operations}\label{sec:Compositions}
We next demonstrate how our lower bounds can be applied to more general computations where
any or all of the following apply:

\begin{enumerate}
\item
We might do a sequence of basic operations (matrix multiplication, LU, etc.).
\item The outputs of one
operation are the inputs to a later one but do not necessarily
need to be saved in slow memory,
\item The inputs may be computed
by formulas (like $A(i,j) = 1/(i+j)$) requiring no memory traffic.
\item The ultimate output written to slow memory may just be a
scalar, like the norm of a matrix.
\item An algorithm might
compute but discard some results rather than save them to memory
(e.g., ILU might discard entries of L or U whose magnitudes falls
below a threshold).
\end{enumerate}

In particular we would like a lower bound  where we are allowed to
arbitrarily interleave all the instructions from all basic
operations in the computation together, and so get a lower bound
for a global optimization of the entire program. For example, if
two different matrix multiplications share a common input matrix,
is it worth trying to interleave instructions from these two
different matrix multiplications?

A natural question is whether it is good enough to just use
optimal implementations of the basic operations, like matrix
multiplication, to attain the global lower bound. This would
clearly be the simplest way to implement the program. We know from
experience that this is not always the case. For example, LU
itself can be decomposed in many ways in terms of operations like
matrix multiplication. Yet only recently have optimal LU
algorithms been constructed. Previous LU algorithms did not attain
optimal bandwidth and latency, even when each of their composing
operations had optimal bandwidth and latency.

We give some examples, such as computing matrix powers, where it
is indeed good enough to use repeated calls to an optimal matrix
multiplication, as opposed to needing a new algorithm, and another
example where the straightforward composition does not suffice,
and a more careful interleaving of the computation is needed in
order to attain the lower bound.

\subsection{The Sequential Case}\label{sec:seq}
\subsubsection{When eliminating input/output does not save much}
In this example we consider a single linear algebra
operation, where inputs are given by formulas and the output is a
scalar (e.g., norm of the product of two matrices given by
formulas, each used once; computing the determinant of a matrix
with entries given by formulas, where one does the $LU$
decomposition and takes the product of the diagonal elements of
$U$, etc.)

Even though this seems to eliminate a large number of reads and
writes, we can prove (for this and similar examples) that the
communication lower bound is still
$\Omega\lt(\frac{\#\text{flops}}{\sqrt M }\rt)$, by using a
technique of {\em imposing reads and writes}: We take an algorithm
to which Theorem~\ref{thm:Main} does {\em not} apply, because it
may potentially have S2/D2 operands, and add (impose) memory
traffic to eliminate such operands. Then we use Theorem
\ref{thm:Main} to bound below the communication of this modified
algorithm, and subtract the amount of imposed communication to get
a lower bound for the original algorithm.

Here is an example. Consider computing $r = \| A \cdot B \|_F^2 =
\sum_{ij} (A \cdot B)_{ij}^2$, where $A_{ik} = 1/(i+k)$ and
$B_{kj}= k^{1/j}$ are given by formulas. Let $C=A \cdot B$.
Whenever the final value of some $C_{ij}$ is computed, squared,
and added to $r$, we impose a write (if it is missing) so that
$C_{ij}$ is saved in slow memory, and so has destination D1
instead of possibly D2 (it may still have source S2). Thus no
entries of $C$ can be S2/D2. Whenever the value of some $A_{ik}$
or $B_{kj}$ is computed by a formula, we impose a read to get it
from a location in slow memory, so it has source S1 instead of S2
(it may still have destination D2). Now, no entries of $A$ or $B$
can be S2/D2. Thus this modified algorithm has lower bound
$n^3/(8\sqrt{M}) - M$ by Theorem~\ref{thm:Main}.

To get a lower bound for the original algorithm, we need to bound how
many reads and writes we imposed. There are clearly at most $n^2$
imposed writes. If the original algorithm only evaluates each formula for
$A_{ik}$ and $B_{kj}$ once, and keeps their computed values in memory
if necessary for later use, then the number of imposed reads is $2n^2$, and the
communication lower bound for the original algorithm is
$n^3/(8\sqrt{M}) - M - 3n^2 = \Omega(n^3/\sqrt{M})$, close to
standard dense matrix multiplication.

On the other hand, if the original algorithm evaluates the formulas for
$A_{ik}$ and $B_{kj}$ whenever it needs them, so $n^3$ times, then
the communication lower bound for the original algorithm becomes
$n^3/(8\sqrt{M}) - M - n^2 - 2n^3$, which degenerates to zero.

\subsubsection{A sequence of basic linear algebra operations}

In the following example, we compose a sequence of basic linear
algebra operations where intermediate outputs are used as inputs
later, and never written to memory (e.g., computing consecutive
powers of a matrix, or repeated squaring). Again, even though this
seems to eliminate a large number of reads and writes, we show
that in some cases the lower bound is still
$\Omega\lt(\frac{\#\text{flops}}{\sqrt M }\rt)$, by imposing reads
and writes and merging all the operations into a single set
satisfying Equation~(\ref{eqn:Model}). This means that in such
cases we can simply call a sequence of individually optimized
linear algebra routines and do asymptotically as well as we would
do with any arbitrary interleaving.

\begin{corollary}[Consecutive powers of a matrix]
\label{cor:repeated} Let $A$ be an $n$-by-$n$ matrix, and let
$Alg$ be a sequential algorithm that computes $A^2 = A \cdot A$,
$A^3 = A^2 \cdot A$, ... , $ A^t = A^{t-1} \cdot A$, but only
needs to save $A^t$ in slow memory. Let $G$ be the total number of
multiplications performed (e.g., $G=(t-1)n^3$ if $A$ is dense),
where we assume that each entry of each $A^i$ is computed at most
once. Then no matter how the operations of $Alg$ are interleaved,
its bandwidth lower bound is $\Omega(\frac{G}{\sqrt{8M}} - M -
(t-2)n^2)$ (if the $A^i$ are sparse, we can subtract less than
$(t-2)n^2$ and get a better lower bound).
\end{corollary}

\begin{proof}
We give two proofs, each of which may be applied to other examples.
For the first proof, we show how all the operations $A^2 = A \cdot A$ , ... ,
$A^t = A^{t-1} \cdot A$, may be combined into one set to which
Equation~(\ref{eqn:Model}), and so Theorem~\ref{thm:Main}, applies.
For Equation~(\ref{eqn:Model}) to apply, we must show that all the inputs, outputs
and multiplications can be indexed by one index set $(i,j,k)$ in the one-to-one
manner described in section~\ref{sec:MainResult}; this is most easily seen by
writing all the operations as
\[
\begin{pmatrix}
  A^2 \\
  A^3 \\
  \vdots \\
  A^t
\end{pmatrix} =
\begin{pmatrix}
  A \\
  A^2 \\
  \vdots \\
  A^{t-1}
\end{pmatrix}
\cdot A
\]
Recall that Equation~(\ref{eqn:Model}) permits inputs and output to overlap,
and ``$a(i,k)$'' and ``$b(k,j)$'' inputs to overlap,
but the ``$a(i,k)$'' inputs alone must be indexed one-to-one,
and similarly the ``$b(k,j)$'' inputs alone must be indexed one-to-one;
this is the case above.

Next, we impose writes of all the intermediate results $A^2, ... , A^{t-1}$,
yielding a new algorithm $Alg'$.
This means that there are no S2/D2 arguments, so Theorem~\ref{thm:Main}
applies to $Alg'$.
Thus the bandwidth lower bound of $Alg'$ is $\frac{G}{\sqrt{8M}} - M$,
and the bandwidth lower bound of $Alg$ is lower by the number of imposed writes,
at most $(t-2)n^2$ (less if the matrices are sparse).

Now we present a second proof, which uses the Loomis-Whitney-based analysis of
a segment more directly.
We let $\#A_i$ be the number of entries of $A^i$ in fast memory during a segment
of $Alg'$.
From the definition of a segment, we can bound
$\sum_{i=1}^t \#A_i \leq 4M$.
Applying Loomis-Whitney to each multiplication $A^{i+1} = A^i \cdot A$ that
one might do (some of) during a segment,
we can bound the number of multiplications during a segment by
$F = \sum_{i=1}^{t-1} \sqrt{\#A_{i+1} \cdot \#A_{i} \cdot \#A_1}$.
We can now bound $F$ subject to the constraint
$\sum_{i=1}^t \#A_i \leq 4M$, yielding
\begin{eqnarray*}
F & = & \sum_{i=1}^{t-1} \sqrt{\#A_{i+1} \cdot \#A_{i} \cdot \#A_1} \\
  & = & \sqrt{\#A_1} \cdot \sum_{i=1}^{t-1} \sqrt{\#A_{i+1} \cdot \#A_{i}} \\
  & \leq & \sqrt{\#A_1} \cdot \sqrt{ \sum_{i=1}^{t-1} \#A_{i+1} } \cdot \sqrt{ \sum_{i=1}^{t-1} \#A_{i} }
  \; \; \; {\rm ...\ by\ the\ Cauchy-Schwarz\ inequality} \\
  & \leq & \sqrt{4M} \cdot \sqrt{4M} \cdot \sqrt{4M} = 8\sqrt{M^3}
\end{eqnarray*}
This yields the ultimate bandwidth lower bound of $G/(8\sqrt{M}) - M$.
\end{proof}

Both proof techniques also apply to repeated squaring:
$A_{i+1} = A_i^2$ for $i=1,...,t-1$, the first proof via the identity
\[
\begin{pmatrix}
  A^2 & & & \\
    & A^4 & & \\
  & & \ddots & \\
  & & & A^{2^{t}}
\end{pmatrix}
=
\begin{pmatrix}
  A & & & \\
    & A^2 & & \\
  & & \ddots & \\
  & & & A^{2^{t-1}}
\end{pmatrix}
\cdot
\begin{pmatrix}
  A & & & \\
    & A^2 & & \\
  & & \ddots & \\
  & & & A^{2^{t-1}}
\end{pmatrix}
\]
and the second proof by bounding the number of multiplications
during a segment by maximizing
$F = \sum_{i=1}^{t-1} \sqrt{\#A_i \cdot \#A_i \cdot \#A_{i+1}}$
subject to $\sum_{i=1}^t \#A_i \leq 4M$ (here $\#A_i$ denotes the number of entries of $A^{2^{i-1}}$ available during a segment).



\subsubsection{Interleaved vs. Phased Sequences of Operations}

In some cases, one can combine and interleave basic linear algebra
operations, (e.g., a sequence of matrix multiplications) so that
the resulting algorithm no longer agrees with
Equation~(\ref{eqn:Model}), although the algorithms for performing
each of the basic linear algebra operations separately do agree
with Equation~(\ref{eqn:Model}). This may lead to an algorithm
whose minimum communication is {\em not} proportional to \#flops,
but asymptotically better.

Before giving an example, we first observe that a ``phased''
algorithm, consisting of a sequence of calls to individually
optimized basic linear algebra operations (like matrix
multiplication), where each such basic linear algebra operation
(phase) must complete before the next can begin, can offer no such
asymptotic improvements. Indeed, if we perform $Alg_1$, ...
,$Alg_t$ in phases, where $Alg_i$ has bandwidth lower bound $B_i$,
then the sequence has bandwidth lower bound $B = \sum_{i=1}^t B_i
- 2(t-1)M$. If each $B_i$ is proportional to the operation count
of $Alg_i$, then $B$ is proportional to the total operation count.
(the modest improvement $2(t-1)M$ arises since we can possibly
avoid a little communication by $Alg_{i+1}$ using the results left
in fast memory by $Alg_i$).

Let us now look at an example, where the interleaved algorithm
can do asymptotically less communication than the phased algorithm:
Consider computing the
dense matrix multiplications $C^{(k)} = A \cdot B^{(k)}$ for
$k=1,2,...,t$ where $B^{(k)}_{i,j} = \sqrt[k]{B_{i,j}}$.

The idea is that having both $A_{i,k}$ and $B_{k,j}$ in fast
memory lets us do up to $t$  evaluations of $g_{ijk}$.
Moreover, the union of all these $tn^3$ operations does not match
Equation~(\ref{eqn:Model}), since the inputs $B_{k,j}$ cannot be
indexed in a one-to-one fashion.
However, we can still give a non-trivial lower bound as follows,
analyzing the algorithm segment by segment.
Let us begin with the lower bound, then show an algorithm
attaining this lower bound.

No operands in a segment are S2/D2.  By the same argument
as in Section~\ref{sec:MainResult},
a maximum of $4M$ arguments of $A$, $B$ and any $C^{(i)}$'s
are available during a segment.
We want to bound the number of $g_{ijk}$'s that we can do during such a segment.
Let $\#A, \#B$ and $\#C^{(i)}$
denote the number of each type of argument available during the
segment. Then by Loomis-Whitney (applied $t$ times) the maximum
number of $g_{ijk}$'s is bounded by
$F = \sum_{i=1}^t \sqrt{\#A \cdot \#B \cdot \#C^{(i)}}$.
We want to maximize $F$ subject to the constraint
  $\#A + \#B + \sum_{i=1}^t \#C^{(i)} \leq 4M$.
Applying Cauchy-Schwarz as before yields
\[
F = \sqrt{\#A} \cdot \sqrt{\#B} \cdot \sum_{i=1}^t \sqrt {\#C^{(i)}}
\leq \sqrt{\#A} \cdot \sqrt{\#B} \cdot \sqrt{ \sum_{i=1}^t \#C^{(i)} } \cdot \sqrt{t}
\leq \sqrt{4M} \cdot \sqrt{4M} \cdot \sqrt{4M} \cdot \sqrt{t} = 8 \sqrt{tM^3} \; \;
\]
The number of segments is thus at least
 $\lt \lfloor \frac{tn^3}{8 M^{3/2} t^{1/2}} \rt \rfloor$
and the number of memory operations at least
  $\frac{t^{1/2}n^3}{8M^{1/2} } -M$.
This is smaller than the ``phased" lower bound for $t$ matrix
multiplications in sequence, $ \frac{tn^3}{8\sqrt M} - tM$, by an
asymptotic factor of $\Theta (\sqrt t)$.

We next show that this bound is indeed attainable, using a
different blocked matrix multiplication algorithm whose block
sizes $b_1$ and $b_2$ depend on $M$ and $t$ (see Algorithm
\ref{alg:mat-mats-mul}). The bandwidth count for this algorithm is
as follows. In the innermost loop we read/write $t$ blocks of
$C^{(1)},....,C^{(t)}$, of $M/3t$ words each. So we have $2M/3$
reads/writes for the innermost loop. Before this loop we read two
blocks (of $A$ and $B$) of $M/3$ words each. This adds up to
$O(M)$ read/writes. This is performed $\frac{n^3}{b_1^2 b_2} $
times. So the total bandwidth count is $O\lt(M \cdot  \lt(
\frac{n^3}{b_1^2 b_2} \rt)\rt) = O \lt( \frac{\sqrt{t}n^3}{\sqrt
M}\rt)$.

\begin{algorithm}
\protect\caption{Matrix-Matrices
multiplication}\label{alg:mat-mats-mul}
\begin{algorithmic}[1]
\label{mat-mats-mul}
\STATE
$b_1 = \sqrt{M/3t}$, $b_2 = \sqrt{Mt/3}$, \{so $b_1 b_2 = M/3$ \}
\STATE
  Break $A$ into blocks of size $b_1 \times b_2$.
\STATE  Break $B$ into blocks of size $b_2 \times b_1$.
\STATE  Break each $C^{(i)}$ into blocks of size $b_1 \times b_1$.
\STATE  Do block matrix multiplication, where the innermost loop reads in a block of $A$,
     a block of $B$, and one block each of $C^{(1)},....,C^{(t)}$, and updates each $C^{(i)}$ :
\FOR{$i=1$ to $n/b_1$}
    \FOR{$j=1$ to $n/b_1$}
        \FOR{$k =1$ to $n/b_2$}
           \STATE Read block $A_{i,k}$ and block $B_{k,j}$
           \FOR{$m = 1$ to $t$}
               \STATE Read block $C^{(m)}_{i,j}$
               \STATE $C^{(m)}_{i,j} += A_{i,k} \cdot (B^{(m)}_{k,j})$  \hspace*{.25in} ...\{$(B^{(m)}_{k,j})$ is recomputed each time \}
               \STATE Write $C^{(m)}_{i,j} $
           \ENDFOR
        \ENDFOR
    \ENDFOR
\ENDFOR
\end{algorithmic}
\end{algorithm}

\subsection{The Parallel Case}

  The techniques in the above Section \ref{sec:seq} for composing sequential linear
  algebra operations can be extended to the parallel case in two different ways.
  When we impose reads and writes to get an algorithm to which our
  previous lower bounds apply, we need to decide which processor's memory
  will participate in those reads and writes. The first option is to create a
  ``twin processor'' for each processor, whose memory will hold this data.
  This doubles the number of processors to which the previous lower bound
  applies, and also requires us to bound the total memory per processor
  not by $NNZ/P$ (again assuming memory is balanced among processors)
  but by the maximum of $NNZ/P$ and the largest number of reads and writes
  imposed on any processor. The second option is to have all the imposed
  reads and writes be in the local processor's memory. This keeps the
  number of processors constant, but increases $NNZ/P$ by adding the
  largest number of imposed reads and writes on each processor.
  The details are algorithm-dependent. For example,
similar to the sequential case, we obtain a tight lower bound for
repeated matrix multiplication and for repeated matrix squaring.

\subsection{Applications to Graph Algorithms}\label{sec:Graphs}

Matrix multiplication algorithms are used to solve many graph
related problems. Thus our lower bounds may hold, as long as the
matrix multiplication algorithm that is used agrees with Equation
(\ref{eqn:Model}). The bounds do not apply when using
Strassen-like algorithm (e.g., \cite{YusterZwick05}).

In some cases, one can directly match the flops performed by an
algorithm to Equation (\ref{eqn:Model}), and obtain a
communication lower bound. We next consider, for example,
matrix-multiplication-like recursive algorithms for finding the
shortest path between any pair of vertices in a graph (the
All-Pairs Shortest-Path problem). For tight upper and lower bounds
for the bandwidth of Floyd-Warshall and other related algorithms,
see \cite{MichaelPennerPrasanna02}. The algorithm works as follows
\cite{CormenLeisersonRivestStein01}. Let $l_{ij}^{(m)}$ be the
minimum weight of any path from vertex $i$ to vertex $j$ that
contains at most $m$ edges, where the weight of the edge $(i,j)$
is $w_{ij}=l_{ij}^{(1)}$. Then $l_{ij}^{(m)} = \min_{1 \leq k \leq
n} \left(l_{ik}^{(m-1)} +w_{kj} \right)$, and the recursive naive
algorithm for the All-Pairs Shortest-Path problems performs
exactly these $\Theta(n^4)$ computations. If all values
$l_{ij}^{(m)}$ are written to slow memory, then, by Theorem
\ref{thm:Main}, the bandwidth lower bound is $\Omega
\left(\frac{n^4}{\sqrt{M}}\right)$. Although this may not
 be the case ---some of the intermediate values may never reach the slow memory---
   there are
fewer than $n^3$ intermediate $l_{ij}^{(m)}$ values. Thus, by
imposing reads and writes,
the bandwidth lower bound is
$\Omega \left( \frac{n^4}{\sqrt{M}}\right)$ (note that here, similar to
the repeated matrix multiplication arguments of
Corollary \ref{cor:repeated},
after imposing writes, no two $g_{ijk}$
operations use the same two inputs, so Equation \ref{eqn:Model}
applies). Similarly, the $\Theta(n^3 \log n)$ recursive algorithm
for APSP has $O(n^2 \log n)$ intermediate values, therefore, by
Theorem \ref{thm:Main} and imposing reads and writes,
the bandwidth lower bound is
 $\Omega\left(\frac{n^3 \log n}{\sqrt{M}} \right)$.

Note that these lower bounds are attainable. As noted before (see
e.g., \cite{CormenLeisersonRivestStein01}) any matrix powering
algorithm can be converted into a APSP algorithm, by using `$+$'
instead of `$*$' and `$\min$' instead of summation. Starting with
any of the communication-avoiding optimal matrix-multiplication
algorithms (e.g., \cite{FrigoLeisersonProkopRamachandran99})
guarantees a bandwidth upper bound of $O
\left(\frac{n^4}{\sqrt{M}} \right)$ and $O\left(\frac{n^3 \log
n}{\sqrt{M}} \right)$ respectively. Using recursive-block data
structure further guarantees optimal latency for both algorithms.

The above repeated-matrix-squaring-like algorithm may,
in some cases, perform better than the communication-avoiding
implementation of Floyd-Warshall algorithm
\cite{MichaelPennerPrasanna02}. Consider the problem of finding
the neighbors of distance $t$ of every vertex.

One can use the above repeated-matrix-squaring-like algorithm for
$\log t$ phases, obtaining a running time of $\Theta(n^3 \log t)$
and communication complexity $\Theta \lt( \frac{n^3 \log t
}{\sqrt{M}} \rt) $ for dense graphs. For sparse input graphs this
may further reduce. For example, when $G$ is a union of cycles and
paths, the running time and communication bandwidth are $O(n^2
 2^t)$ and $O\lt( \frac{n^2 2^t }{\sqrt{M}} \rt) $ (as the degree
 of a vertex of the $i$th phase is at most $2^{2^i}$).

If, however, we use the Floyd-Warshall algorithm for this purpose,
we have to run it all the way through, regardless of the input
graph, resulting in running time of $\Theta(n^3)$ and
communication complexity of $\Theta \lt( \frac{n^3}{\sqrt{M}} \rt)
$ (assuming the above communication-avoiding implementation).
Thus, for $t=o( \log n)$ the repeated-matrix-squaring-like
algorithm performs better for constant-degree inputs, both from
flops count and from communication bandwidth perspectives.


\section{Attaining the lower bounds, and open problems}\label{sec:open}

A major problem is to find algorithms that attain the lower
bounds described in this paper, for the various linear algebra
problems, for dense and sparse matrices, and for sequential and
parallel machines. Tables \ref{tbl:seq} and \ref{tbl:par}
summarize the current state-of-the-art (to the best of our
knowledge) for the communication complexity of dense algorithms.
Briefly, all the lower bounds are
attainable in the dense sequential case (Table~\ref{tbl:seq}),
and in the dense parallel case (Table~\ref{tbl:par},
assuming minimal memory $O(n^2/P)$ per processor, and modulo ${\rm polylog} P$ terms).
However, only a few of these algorithms appear in standard
libraries like LAPACK \cite{LAPACK} and ScaLAPACK \cite{SCALAPACK};
the complexity of ScaLAPACK implementations is taken from
\cite[Table 5.8]{SCALAPACK}. Other libraries may well attain similar bounds
\cite{GunnelsGustavsonHenryGeijn01,PLAPACK}.

\begin{table}[!h]
\centering
\begin{tabular}{|l||c|c||c c|} \hline
 & \multicolumn{2}{|c||}{Lower bound} & \multicolumn{2}{|c|}{Upper bound} \\
\hline
Algorithm& Bandwidth & Latency & Bandwidth & Latency \\
\hline \hline
Matrix-Multiplication & & & $O \lt(\frac{n^3}{\sqrt M} \rt) $ & $O \lt(\frac{n^3}{ M^{3/2}} \rt) $ \\
                      & & &\multicolumn{2}{c|}{\cite{FrigoLeisersonProkopRamachandran99}}\\
\cline{1-1} \cline{4-5}
Cholesky              & & & $O \lt(\frac{n^3}{\sqrt M} \rt) $&$O \lt(\frac{n^3}{ M^{3/2}} \rt) $ \\
                      & & &\multicolumn{2}{c|}{\cite{AhmedPingali00,BallardDemmelHoltzSchwartz09}}\\
\cline{1-1} \cline{4-5}
$LU$                  & & & $O \lt(\frac{n^3}{\sqrt M} \rt) $& $O \lt(\frac{n^3}{ M^{3/2}} \rt) $ \\
                      & & &\cite{Toledo97} & \cite{DemmelGrigoriXiang08}\\
\cline{1-1} \cline{4-5}
$QR$                  & $\Omega \lt(\frac{n^3}{\sqrt M} \rt) $&$\Omega \lt(\frac{n^3}{ M^{3/2}} \rt) $
                          & $O \lt(\frac{n^3}{\sqrt M} \rt) $& $O \lt(\frac{n^3}{ M^{3/2}} \rt) $ \\
                                                & & & \cite{ElmrothGustavson98} & \cite{DemmelGrigoriHoemmenLangou08a}\\
\cline{1-1} \cline{4-5}
Symmetric Eigenvalues           & & & $O(\frac{n^3}{\sqrt{M}})$ & $O\lt(\frac{n^3}{M^{3/2}} \rt)$ \ \\
                      & & &\multicolumn{2}{c|}{\cite{BallardDemmelDumitriu09}}\\
\cline{1-1} \cline{4-5}
SVD                   & & & $O(\frac{n^3}{\sqrt{M}})$ & $O\lt(\frac{n^3}{M^{3/2}} \rt)$ \ \\
                      & & &\multicolumn{2}{c|}{\cite{BallardDemmelDumitriu09}}\\
\cline{1-1} \cline{4-5}
(Generalized) Nonsymmetric & & & $O(\frac{n^3}{\sqrt{M}})$ & $O\lt(\frac{n^3}{M^{3/2}} \rt)$ \ \\
Eigenvalues                & & &\multicolumn{2}{c|}{\cite{BallardDemmelDumitriu09}}\\
\hline
\end{tabular}
  \protect\caption{Sequential $\Theta(n^3)$ algorithms: bandwidth and latency
lower bound vs. upper bounds. $M$ denotes the size of the fast
memory.
}\label{tbl:seq}
\end{table}

\begin{table}[!h]
  \centering
\begin{tabular}{|l||c|c||c c|} \hline
 & \multicolumn{2}{|c||}{Lower bound} & \multicolumn{2}{|c|}{Upper bound} \\
\hline
Algorithm& Bandwidth & Latency & Bandwidth & Latency \\
\hline \hline
Matrix-Multiplication & & & $O \lt(\frac{n^2 }{\sqrt P} \rt) $ & $O \lt(\sqrt P  \rt) $ \\
                      & & &\multicolumn{2}{c|}{\cite{Cannon69}}\\
\cline{1-1} \cline{4-5}
Cholesky              & & & $O \lt(\frac{n^2 \log P}{\sqrt P} \rt) $ & $O \lt(\sqrt P \log P \rt) $ \\
                      & & &\multicolumn{2}{c|}{\cite{SCALAPACK}}\\
\cline{1-1} \cline{4-5}
$LU$                  & $\Omega \lt(\frac{n^3}{P\sqrt M} \rt) $&$\Omega \lt(\frac{n^3}{ PM^{3/2}} \rt) $ & $O \lt(\frac{n^2 \log P}{\sqrt P} \rt) $ & $O \lt(\sqrt P \log P \rt) $ \\
                      & & &\cite{DemmelGrigoriXiang08} & \cite{DemmelGrigoriXiang08}\\
\cline{1-1} \cline{4-5}
$QR$                  & $= \Omega \lt(\frac{n^2}{\sqrt P} \rt) $ & $= \Omega \lt(\sqrt P \rt) $
                      & $O \lt(\frac{n^2 \log P}{\sqrt P} \rt) $ & $O \lt(\sqrt P \log^3 P \rt) $ \\
                                                & & & \cite{DemmelGrigoriHoemmenLangou08a}& \cite{DemmelGrigoriHoemmenLangou08a}\\
\cline{1-1} \cline{4-5}
Symmetric Eigenvalues & & & $O \lt(\frac{n^2 \log P}{\sqrt P} \rt) $ & $O(\sqrt P \log^3 P)$ \\
                      & & &\multicolumn{2}{c|}{\cite{BallardDemmelDumitriu09}}\\
\cline{1-1} \cline{4-5}
SVD                   & & & $O \lt(\frac{n^2 \log P}{\sqrt P} \rt) $ & $O(\sqrt P \log^3 P)$ \\
                      & & &\multicolumn{2}{c|}{\cite{BallardDemmelDumitriu09}}\\
\cline{1-1} \cline{4-5}
(Generalized) Nonsymmetric & & & $O \lt(\frac{n^2 \log P}{\sqrt P} \rt) $ & $O(\sqrt P \log^3 P)$ \\
Eigenvalues           & & &\multicolumn{2}{c|}{\cite{BallardDemmelDumitriu09}}\\
\cline{1-1} \cline{4-5}
\hline
\end{tabular}
  \protect\caption{Parallel $\Theta \lt(\frac{n^3}{P} \rt)$ flops algorithms:
bandwidth and latency lower bound vs. upper bounds.
$P$ denotes the number of processors,
and $M=\Theta\lt(\frac{n^2}{P} \rt)$ denotes the size of the memory per processor
(this is the smallest possible memory per processor).
  }\label{tbl:par}
\end{table}

Best understood are dense matrix-multiplication, other BLAS routines,
and Cholesky, which have algorithms that attain (perhaps modulo
${\rm polylog} P$ factors) both bandwidth and latency lower bounds on
parallel machines, and on sequential machines with multiple levels
of memory hierarchy.
The optimal sequential Cholesky algorithm cited in Table~\ref{tbl:seq}
was presented in \cite{AhmedPingali00}, but first analyzed later in
\cite{BallardDemmelHoltzSchwartz09}.
The complexity of ScaLAPACK's parallel Cholesky cited in Table~\ref{tbl:par}
assumes that the largest possible block size is chosen
($NB \approx N/\sqrt{P}$ in line ``PxPOSV'' in \cite[Table 5.8]{SCALAPACK}).

More recently, optimal dense LU and QR algorithms have been
proposed that attain both bandwidth and latency lower bounds in
parallel or sequentially (with just 2 levels of memory hierarchy).
Interestingly, conventional partial pivoting must apparently be
abandoned in order to minimize both latency and bandwidth in LU
\cite{DemmelGrigoriXiang08}; we can retain partial pivoting if we
only want to minimize bandwidth \cite{Toledo97}. Similarly, we
must apparently change the standard representation of the Q matrix
in QR in order to minimize both latency and bandwidth
\cite{DemmelGrigoriHoemmenLangou08a}; we can retain the usual
representation if we only want to minimize bandwidth
\cite{ElmrothGustavson98}. See the above references for large
speedups reported over algorithms that do not try to minimize
communication. The ideas behind communication-optimal dense QR
first appear in \cite{GolubPlemmonsSameh88}, and include
\cite{ButtariLangouKurzakDongarra07,GunterVanDeGeijn05,ElmrothGustavson98};
see \cite{DemmelGrigoriHoemmenLangou08a} for a more complete list
of references.

ScaLAPACK's parallel symmetric eigensolver and SVD routine also
minimize bandwidth (modulo a $\log P$ factor), but not latency,
sending $O(n)$ messages. ScaLAPACK's nonsymmetric eigensolver communicates
much more, indeed just the Hessenberg QR iteration has $n$-times higher bandwidth.
LAPACK's symmetric and nonsymmetric eigensolvers and SVD minimize neither
bandwidth nor latency, with $O(n^3)$ bandwidth.
Recently proposed algorithms in
\cite{BallardDemmelDumitriu09,DemmelDumitriuHoltz07}
for the symmetric and nonsymmetric eigenproblems,
generalized nonsymmetric eigenproblems and SVD do appear to
attain the desired communication complexity
(modulo ${\rm polylog} P$ factors) but at the cost
of doing a possibly large constant factor more arithmetic.
(This is in contrast to the new dense LU and QR algorithms,
which do at most $O(n^2)$ more arithmetic than their
conventional counterparts.)
We note that our lower bound in Section~\ref{sec:QR} does not
apply to the first phase of the conventional algorithm for
the generalized nonsymmetric eigenproblem (reducing the pair $(A,B)$
to (Hessenberg,triangular) form), but we conjecture that it can be
extended to do so. The lower bound does apply for the generalized
nonsymmetric eigenvalue algorithm in
\cite{BallardDemmelDumitriu09,DemmelDumitriuHoltz07}.

Otherwise, for $LDL^T$, for ``3D'' parallel algorithms
other than matrix-multiplication \cite{IronyToledoTiskin04}
that replicate data and so use more memory in order to reduce communication,
for Strassen-like algorithms,
and for sparse matrices in general,
the problems are open.

We note that for sufficiently rectangular dense matrices (e.g.,
matrix-vector multiplication) or for sufficiently sparse matrices,
our lower bound may be lower than the trivial lower bound
(\#inputs + \#outputs) and so not be attainable. In this case the
natural question is whether the maximum of the two lower bounds is
attainable (as it is for dense matrix multiplication).

\newpage

\bibliographystyle{alpha}
\bibliography{bib_glb_arxiv_29}

\newcommand{\etalchar}[1]{$^{#1}$}
\begin{thebibliography}{GGHvdG01}

\bibitem[ABB{\etalchar{+}}92]{LAPACK}
E.~Anderson, Z.~Bai, C.~Bischof, J.~Demmel, J.~Dongarra, J.~Du Croz,
  A.~Greenbaum, S.~Hammarling, A.~McKenney, S.~Ostrouchov, and D.~Sorensen.
\newblock {\em LAPACK's user's guide}.
\newblock Society for Industrial and Applied Mathematics, Philadelphia, PA,
  USA, 1992.
\newblock Also available from http://www.netlib.org/lapack/.

\bibitem[AGW01]{AndersenGustavsonWasniewski01}
B.~S. Andersen, F.~Gustavson, and J.~Wasniewski.
\newblock A recursive formulation of cholesky factorization of a matrix in
  packed storage format.
\newblock {\em {ACM} Transactions on Mathematical Software}, 27(2):214--244,
  jun 2001.

\bibitem[AP00]{AhmedPingali00}
N.~Ahmed and K.~Pingali.
\newblock Automatic generation of block-recursive codes.
\newblock In {\em Euro-Par '00: Proceedings from the 6th International Euro-Par
  Conference on Parallel Processing}, pages 368--378, London, UK, 2000.
  Springer-Verlag.

\bibitem[AV88]{AggarwalVitter88}
A.~Aggarwal and J.~S. Vitter.
\newblock The input/output complexity of sorting and related problems.
\newblock {\em Commun. ACM}, 31(9):1116--1127, 1988.

\bibitem[BBF{\etalchar{+}}07]{BenderBrodalFagerbergJacobVicari07}
M.~A. Bender, G.~S. Brodal, R.~Fagerberg, R.~Jacob, and E.~Vicari.
\newblock Optimal sparse matrix dense vector multiplication in the {I/O}-model.
\newblock In {\em SPAA '07: Proceedings of the nineteenth annual ACM symposium
  on Parallel algorithms and architectures}, pages 61--70, New York, NY, USA,
  2007. ACM.

\bibitem[BBM02a]{BramanByersMathias02a}
K.~Braman, R.~Byers, and R.~Mathias.
\newblock The {M}ulti-{S}hift {QR} {A}lgorithm, {P}art {I}: {M}aintaining
  {W}ell {F}ocused {S}hifts and {L}evel 3 {P}erformance.
\newblock {\em {SIAM} {J.} {M}atrix {A}nal. {A}pp.}, 23(4):929--947, 2002.

\bibitem[BBM02b]{BramanByersMathias02b}
K.~Braman, R.~Byers, and R.~Mathias.
\newblock The {M}ulti-{S}hift {QR} {A}lgorithm, {P}art {II}: {A}ggressive
  {E}arly {D}eflation.
\newblock {\em {SIAM} {J.} {M}atrix {A}nal. {A}pp.}, 23(4):948--973, 2002.

\bibitem[BCC{\etalchar{+}}97]{SCALAPACK}
L.~S. Blackford, J.~Choi, A.~Cleary, E.~{D'Azevedo}, J.~Demmel, I.~Dhillon,
  J.~Dongarra, S.~Hammarling, G.~Henry, A.~Petitet, K.~Stanley, D.~Walker, and
  R.~C. Whaley.
\newblock {\em {ScaLAPACK} Users' Guide}.
\newblock SIAM, Philadelphia, PA, USA, May 1997.
\newblock Also available from http://www.netlib.org/scalapack/.

\bibitem[BDD{\etalchar{+}}01]{newBLASfull}
L.~S. Blackford, J.~Demmel, J.~Dongarra, I.~Duff, S.~Hammarling, G.~Henry,
  M.~Heroux, L.~Kaufman, A.~Lumsdaine, A.~Petitet, R.~Pozo, K.~Remington, R.~C.
  Whaley, Z.~Maany, F.~Krough, G.~Corliss, C.~Hu, B.~Keafott, W.~Walster, and
  J.~Wolff~v. Gudenberg.
\newblock {B}asic {L}inear {A}lgebra {S}ubprograms {T}echical ({BLAST}) {F}orum
  {S}tandard.
\newblock {\em Intern. J. High Performance Comput.}, 15(3-4), 2001.

\bibitem[BDD{\etalchar{+}}02]{newBLAS}
L.~S. Blackford, J.~Demmel, J.~Dongarra, I.~Duff, S.~Hammarling, G.~Henry,
  M.~Heroux, L.~Kaufman, A.~Lumsdaine, A.~Petitet, R.~Pozo, K.~Remington, and
  R.~C. Whaley.
\newblock An updated set of {B}asic {L}inear {A}lgebra {S}ubroutines ({BLAS}).
\newblock {\em ACM Trans. Math. Soft.}, 28(2), June 2002.

\bibitem[BDD09]{BallardDemmelDumitriu09}
G.~Ballard, J.~Demmel, and I.~Dumitriu.
\newblock Communication-optimal parallel and sequential eigenvalue and singular
  value algorithms.
\newblock In preparation, 2009.

\bibitem[BDHS09]{BallardDemmelHoltzSchwartz09}
G.~Ballard, J.~Demmel, O.~Holtz, and O.~Schwartz.
\newblock {Communication-optimal Parallel and Sequential Cholesky
  Decomposition}, 2009.
\newblock Submitted. Also available in EECS Tech Report and from the ar{X}iv:
  http://arxiv.org/abs/0902.2537.

\bibitem[BLA]{BLAS}
{BLAS} - {B}asic {L}inear {A}lgebra {S}ubroutines.
\newblock www.netlib.org/blas.

\bibitem[BLKD07]{ButtariLangouKurzakDongarra07}
A.~Buttari, J.~Langou, J.~Kurzak, and J.~J. Dongarra.
\newblock A class of parallel tiled linear algebra algorithms for multicore
  architectures.
\newblock Technical Report 191, LAPACK Working Note, September 2007.

\bibitem[BVL87]{BischofVanLoan87}
C.~Bischof and C.~Van~Loan.
\newblock The {WY} representation for products of {H}ouseholder matrices.
\newblock {\em {SIAM} J. Sci. Stat. Comp.}, 8(1), 1987.

\bibitem[BZ88]{BuragoZalgaller88}
Y.~D. Burago and V.~A. Zalgaller.
\newblock {\em Geometric Inequalities}, volume 285 of {\em Grundlehren der
  Mathematische Wissenschaften}.
\newblock Springer, Berlin, 1988.

\bibitem[Can69]{Cannon69}
L.~Cannon.
\newblock {\em A cellular computer to implement the Kalman filter algorithm}.
\newblock PhD thesis, Montana State University, Bozeman, MN, 1969.

\bibitem[CLRS01]{CormenLeisersonRivestStein01}
Thomas~H. Cormen, Charles~E. Leiserson, Ronald~L. Rivest, and Clifford Stein.
\newblock {\em Introduction to Algorithms}.
\newblock The MIT Press, 2nd edition, 2001.

\bibitem[CR06]{ChowdhuryRamachandran06}
R.~A. Chowdhury and V.~Ramachandran.
\newblock Cache-oblivious dynamic programming.
\newblock In {\em SODA '06: Proceedings of the seventeenth annual ACM-SIAM
  symposium on Discrete algorithm}, pages 591--600, New York, NY, USA, 2006.
  ACM.

\bibitem[DDH07]{DemmelDumitriuHoltz07}
J.~Demmel, I.~Dumitriu, and O.~Holtz.
\newblock Fast linear algebra is stable.
\newblock {\em Numerische Mathematik}, 108(1):59--91, 2007.

\bibitem[Dem96]{CS267a}
J.~Demmel.
\newblock {CS} 267 {C}ourse {N}otes: {A}pplications of {P}arallel {P}rocessing.
\newblock Computer Science Division, University of California, 1996.
\newblock http://www.cs.berkeley.edu/$\sim$demmel/cs267.

\bibitem[Dem97]{Demmel97}
J.~Demmel.
\newblock {\em Applied Numerical Linear Algebra}.
\newblock SIAM, 1997.

\bibitem[DGHL08a]{DemmelGrigoriHoemmenLangou08a}
J.~Demmel, L.~Grigori, M.~Hoemmen, and J.~Langou.
\newblock Communication-optimal parallel and sequential {QR} and {LU}
  factorizations.
\newblock UC Berkeley Technical Report EECS-2008-89, Aug 1, 2008; Submitted to
  SIAM. J. Sci. Comp., 2008.

\bibitem[DGHL08b]{DemmelGrigoriHoemmenLangou08b}
J.~Demmel, L.~Grigori, M.~Hoemmen, and J.~Langou.
\newblock Implementing communication-optimal parallel and sequential {QR} and
  {LU} factorizations.
\newblock submitted to SIAM. J. Sci. Comp., 2008.

\bibitem[DGX08]{DemmelGrigoriXiang08}
J.~Demmel, L.~Grigori, and H.~Xiang.
\newblock Communication-avoiding {Gaussian} elimination.
\newblock Supercomputing 08, 2008.

\bibitem[DP04]{DhillonParlett04}
I.~Dhillon and B.~Parlett.
\newblock Orthogonal {E}igenvectors and {R}elative {G}aps.
\newblock {\em {SIAM} {J.} on {M}atrix {A}nal. {A}pp.}, 25(3):858--899, 2004.

\bibitem[DPV06]{DhillonParlettVoemel06}
I.~Dhillon, B.~Parlett, and C.~V\"{o}mel.
\newblock The {D}esign and {I}mplementation of the {MRRR} {A}lgorithm.
\newblock {\em {ACM} {T}rans. {M}ath. {S}oft.}, 32(4):533--560, 2006.

\bibitem[EG98]{ElmrothGustavson98}
E.~Elmroth and F.~Gustavson.
\newblock New serial and parallel recursive {QR} factorization algorithms for
  {SMP} systems.
\newblock In B.~K{\aa}gstr{\"o}m et~al., editor, {\em Applied Parallel
  Computing. Large Scale Scientific and Industrial Problems.}, volume 1541 of
  {\em Lecture Notes in Computer Science}, pages 120--128. Springer, 1998.

\bibitem[EG00]{Elmroth00}
E.~Elmroth and F.~Gustavson.
\newblock Applying recursion to serial and parallel {QR} factorization leads to
  better performance.
\newblock {\em IBM Journal of Research and Development}, 44(4):605--624, 2000.

\bibitem[EGJK04]{ElmrothGustavsonJonssonKagstrom04}
E.~Elmroth, F.~Gustavson, I.~Jonsson, and B.~K{\aa}gstr{\"o}m.
\newblock Recursive blocked algorithms and hybrid data structures for dense
  matrix library software.
\newblock {\em {SIAM} {R}eview}, 46(1):3--45, March 2004.

\bibitem[FLPR99]{FrigoLeisersonProkopRamachandran99}
M.~Frigo, C.~E. Leiserson, H.~Prokop, and S.~Ramachandran.
\newblock Cache-oblivious algorithms.
\newblock In {\em FOCS '99: Proceedings of the 40th Annual Symposium on
  Foundations of Computer Science}, page 285, Washington, DC, USA, 1999. IEEE
  Computer Society.

\bibitem[Geo73]{George73}
A.~George.
\newblock Nested dissection of a regular finite element mesh.
\newblock {\em {SIAM} J. Numer. Anal.}, 10:345--363, 1973.

\bibitem[GG05]{GunterVanDeGeijn05}
B.~C. Gunter and R.~A. Van~De Geijn.
\newblock Parallel out-of-core computation and updating of the {QR}
  factorization.
\newblock {\em ACM Trans. Math. Softw.}, 31(1):60--78, 2005.

\bibitem[GGHvdG01]{GunnelsGustavsonHenryGeijn01}
J.~A. Gunnels, F.~G. Gustavson, G.~M. Henry, and R.~A. van~de Geijn.
\newblock {FLAME}: {Formal Linear Algebra Methods Environment}.
\newblock {\em {ACM} Transactions on Mathematical Software}, 27(4):422--455,
  December 2001.

\bibitem[GPS88]{GolubPlemmonsSameh88}
G.~H. Golub, R.~J. Plemmons, and A.~Sameh.
\newblock Parallel block schemes for large-scale least-squares computations.
\newblock pages 171--179, 1988.

\bibitem[GT87]{GilbertTarjan87}
J.~R. Gilbert and R.~E. Tarjan.
\newblock The analysis of a nested dissection algorithm.
\newblock {\em Numerische Mathematik}, pages 377--404, 1987.

\bibitem[GVL96]{GolubVanloan96}
G.~Golub and C.~Van~Loan.
\newblock {\em Matrix Computations}.
\newblock Johns Hopkins University Press, Baltimore, MD, 3rd edition, 1996.

\bibitem[HK81]{HongKung81}
J.~W. Hong and H.~T. Kung.
\newblock {I/O} complexity: The red-blue pebble game.
\newblock In {\em STOC '81: Proceedings of the thirteenth annual ACM symposium
  on Theory of computing}, pages 326--333, New York, NY, USA, 1981. ACM.

\bibitem[HMR73]{HMR73}
A.~J. Hoffman, M.~S. Martin, and D.~J. Rose.
\newblock Complexity bounds for regular finite difference and finite element
  grids.
\newblock {\em {SIAM} J. Numer. Anal.}, 10:364--369, 1973.

\bibitem[ITT04]{IronyToledoTiskin04}
D.~Irony, S.~Toledo, and A.~Tiskin.
\newblock Communication lower bounds for distributed-memory matrix
  multiplication.
\newblock {\em J. Parallel Distrib. Comput.}, 64(9):1017--1026, 2004.

\bibitem[LW49]{LoomisWhitney49}
L.~H. Loomis and H.~Whitney.
\newblock An inequality related to the isoperimetric inequality.
\newblock {\em Bulletin of the {AMS}}, 55:961--962, 1949.

\bibitem[MPP02]{MichaelPennerPrasanna02}
J.~P. Michael, M.~Penner, and V.~K. Prasanna.
\newblock Optimizing graph algorithms for improved cache performance.
\newblock In {\em In Proc. Int'l Parallel and Distributed Processing Symp.
  (IPDPS 2002), Fort Lauderdale, FL}, pages 769--782, 2002.

\bibitem[Pug92]{Puglisi92}
C.~Puglisi.
\newblock Modification of the {H}ouseholder method based on compact {WY}
  representation.
\newblock {\em {SIAM} J. Sci. Stat. Comput.}, 13(3):723--726, 1992.

\bibitem[Saa96]{Saad96}
Y.~Saad.
\newblock {\em Iterative {M}ethods for {S}parse {L}inear {S}ystems}.
\newblock {PWS} {P}ublishing {C}o., Boston, 1996.

\bibitem[Sav95]{Savage95}
J.~E. Savage.
\newblock Extending the {Hong-Kung} model to memory hierarchies.
\newblock In {\em COCOON}, pages 270--281, 1995.

\bibitem[SVL89]{Schreiber87a}
R.~Schreiber and C.~Van~Loan.
\newblock A storage efficient {WY} representation for products of {H}ouseholder
  transformations.
\newblock {\em {SIAM} J. Sci. Stat. Comput.}, 10:53--57, 1989.

\bibitem[Tol97]{Toledo97}
S.~Toledo.
\newblock Locality of reference in {LU} decomposition with partial pivoting.
\newblock {\em SIAM J. Matrix Anal. Appl.}, 18(4):1065--1081, 1997.

\bibitem[vdG]{PLAPACK}
R.~van~de Geijn.
\newblock {PLAPACK}: {P}arallel {L}inear {A}lgebra {P}ackage.
\newblock www.cs.utexas.edu/users/plapack.

\bibitem[VDY05]{VuducDemmelYelick05}
R.~Vuduc, J.~Demmel, and K.~Yelick.
\newblock {OSKI}: {A} library of automatically tuned sparse matrix kernels.
\newblock In {\em Proc. of {SciDAC} 2005, J. of Physics: Conference Series}.
  Institute of Physics Publishing, June 2005.

\bibitem[YZ05]{YusterZwick05}
R.~Yuster and U.~Zwick.
\newblock Fast sparse matrix multiplication.
\newblock {\em ACM Trans. Algorithms}, 1(1):2--13, 2005.

\end{thebibliography}

\bigskip

\noindent
{\bf Grey Ballard}\\  \noindent
ballard@eecs.berkeley.edu\\ \noindent
Computer Science Division,
University of California, Berkeley, CA 94720. Research supported
by Microsoft and Intel funding (Award $\#$20080469) and by
matching funding by U.C. Discovery (Award $\#$DIG07-10227).

\medskip

\noindent
{\bf James Demmel}\\ \noindent
demmel@cs.berkeley.edu\\ \noindent
Mathematics Department and Computer Science Division,
University of California, Berkeley, CA 94720.
\medskip

\noindent
{\bf Olga Holtz}\\ \noindent
oholtz@eecs.berkeley.edu \\ \noindent
Departments of Mathematics,
University of California, Berkeley and Technische Universit\"at
Berlin. O. Holtz acknowledges support of the Sofja Kovalevskaja
programm of Alexander von Humboldt Foundation.

\medskip

\noindent
{\bf Oded Schwartz} \\ \noindent
odedsc@math.tu-berlin.de \\ \noindent
Departments of Mathematics,
Technische Universit\"at Berlin, 10623 Berlin, Germany.

\end{document}